\documentclass[10pt]{article}

\usepackage{url}
\usepackage[authoryear, round]{natbib}
\usepackage{mathtools}
\usepackage{amsmath}
\usepackage{amssymb}
\usepackage{mathabx}
\usepackage{amsthm}
\usepackage{empheq}
\usepackage{latexsym}
\usepackage{enumitem}
\usepackage{eurosym}
\usepackage{dsfont}
\usepackage{appendix}
\usepackage{color} 

\usepackage[unicode]{hyperref}
\usepackage[nameinlink]{cleveref}

\usepackage{frcursive}
\usepackage[utf8]{inputenc}
\usepackage[T1]{fontenc}
\usepackage{multirow}
\usepackage{lmodern}
\usepackage{anyfontsize}
\usepackage{stmaryrd}
\usepackage{cases}
\usepackage{calc}
\usepackage{pdfpages}
\usepackage{float}
\usepackage[english]{babel}
\usepackage[english=british]{csquotes}

\definecolor{red}{rgb}{0.7,0.15,0.15}
\definecolor{green}{rgb}{0,0.5,0}
\definecolor{blue}{rgb}{0,0,0.7}
\hypersetup{colorlinks, linkcolor={red},citecolor={green}, urlcolor={blue}}
\makeatletter \@addtoreset{equation}{section}

\usepackage[justification=centering]{caption}
\newtheorem{theorem}{Theorem}[section]
\newtheorem{assumption}[theorem]{Assumption}
\newtheorem{corollary}[theorem]{Corollary}

\newtheorem{lemma}[theorem]{Lemma}
\newtheorem{proposition}[theorem]{Proposition}

\newtheorem{definition}[theorem]{Definition}
\newtheorem{remark}[theorem]{Remark}

\makeatletter
\renewenvironment{proof}[1][\relax]{\par
\pushQED{\qed}%
\normalfont \topsep6\p@\@plus6\p@\relax
\trivlist
\item[\hskip\labelsep\itshape
\ifx#1\relax \proofname\else\proofname{} of #1\fi\@addpunct{. }]\ignorespaces
}{%
\popQED\endtrivlist\@endpefalse
}
\makeatother

\usepackage{indentfirst}

\usepackage[singlespacing]{setspace}
\usepackage[hmargin=0.7in,vmargin=0.9in]{geometry}

\DeclareUnicodeCharacter{014D}{\=o}
\def \E{\mathbb{E}}

\def \P{\mathbb{P}}

\def \R{\mathbb{R}}

\def\Ac{{\cal A}}

\def\Cc{{\cal C}}

\def\Qc{{\cal Q}}

\def\erm{\mathrm{e}}
\def\drm{\mathrm{d}}
\newcommand*{\email}[1]{\href{mailto:#1}{\nolinkurl{#1}} }

\usepackage{authblk}
\title{Optimal contracts under adverse selection for staple goods: efficiency of in-kind insurance\footnote{The authors gratefully acknowledge the support of the ANR project PACMAN ANR-16-CE05-0027.}}

\author[1]{Clémence {\sc Alasseur}}
\author[1]{Corinne {\sc Chaton}}
\author[2]{Emma {\sc Hubert}}

\affil[1]{Finance for Energy Market Research Centre (FIME), Paris, France.}
\affil[2]{Imperial College London, Department of Mathematics, London SW7 1NE, United--Kingdom.}

\date{\today}

\begin{document}

\maketitle

\begin{abstract}
An income loss can have a negative impact on households, forcing them to reduce their consumption of some staple goods. This can lead to health issues and, consequently, generate significant costs for society. We suggest that consumers can, to prevent these negative consequences, buy insurance to secure sufficient consumption of a staple good if they lose part of their income. We develop a two-period/two-good principal-agent problem with adverse selection and endogenous reservation utility to model insurance with in-kind benefits. This model allows us to obtain semi-explicit solutions for the insurance contract and is applied to the context of fuel poverty. For this application, our model allows to conclude that, even in the least efficient scenario from the households point of view, \textit{i.e.}, when the insurance is provided by a monopoly, this mechanism decreases significantly the risk of fuel poverty of households by ensuring them a sufficient consumption of energy. The effectiveness of in-kind insurance is highlighted through a comparison with income insurance, but our results nevertheless underline the need to regulate such insurance market.

\medskip



\end{abstract}

\section{Introduction}\label{sec:introduction}

This article examines a solution to protect vulnerable households from the risk of temporary poverty in a particular staple good.
Staple goods are essential products that consumers are unable or unwilling to cut out of their budgets, such as food and water.
They tend to consume them at a relatively constant level, regardless of their financial situation or the price of goods. However, in the event of a substantial loss of income, consumers will reduce the budget allocated to \textit{some} staple goods, such as fresh food, energy, medicines and feminine hygiene products, which can lead to serious illnesses. Indeed, these particular staple goods, unlike food or water, may not seem essential to poor consumers, who, faced with the choice to heat, heal or eat, will systematically make the choice to eat. This question has been the subject of many surveys and discussions on this type of good, which, for households in a very precarious situation, can be considered, in a sense, as a luxury. In particular, we can mention the works of \citet{milne1991effect, freeman2003health} on health care and \citet{meier2013necessity} or \citet{schulte2017price} on energy. We consider in this paper this particular type of staple goods, the consumption of which is strongly impacted by income losses.

\medskip

We choose to focus on modest and vulnerable households, especially those that are not used to saving money in anticipation of the future. An income loss would thus force them to reduce their consumption of some particular basic commodities, such as residential heating, which could lead to health issues (see \citet{lacroix2015fuel}).
A common solution is to offer a financial assistance, \textit{i.e.}, to provide money to the household in the event of income loss. However, in the situation under consideration, the household will spend the money received to buy products that are more essential from its point of view. This type of mechanism thus appears not to be an appropriate solution to the problem we are confronted with. Therefore, our suggestion is in-kind support rather than financial assistance: if the household suffers from a loss of income, it receives a specified quantity of the particular good under consideration, ensuring adequate consumption of that good. This approach, which could be perceived as paternalistic, should protect individuals from the negative effects of some of their decisions and thus prevent the health issues associated with the lack of this good. Indeed, many authors, such as \cite{blackorby1988cash} or \cite{slesnick1996consumption}, show the effectiveness of in-kind support compared to financial support in fighting poverty. 

\medskip

This paternalistic attitude may be judged favourably by the State, as demonstrated, for example, by the introduction of the energy voucher in France in 2018 to fight fuel poverty.\footnote{The reader is referred to \Cref{ss:fuel_poverty} for more details on fuel poverty.}
Unfortunately, it seems that the government cannot get involved to help all households in difficulty, nor is it able to offer vouchers such as the energy voucher for all other essential goods. The purpose of this paper is therefore to determine what type of mechanism could be proposed by the private sector (insurer, supplier, producer...), so that it is profitable for it, and so that it has the effects desired by the State, \textit{i.e.}, to enable households to sufficiently consume the good in question despite a loss of income. In particular, we will study an insurance mechanism, provided by the private sector and dedicated to a particular good, that could guarantee the insured household sufficient consumption of this product in the event of income loss.

\medskip

Since the risk of income loss is different among the considered population, there is a need to offer a menu of contracts: different types of insurance should be provided to let each household choose the contract best suited to its needs and risk.
Offering a menu of contracts is relatively traditional in the insurance field (see the survey of \citet{dionne2013adverse}), since it allows insurers to partially remedy information asymmetries between them and the policyholders, in particular, \textit{adverse selection}. 
In our context, this asymmetry is related to the risk of income loss to which the insured is exposed. Indeed, it seems reasonable to assume that the insured knows more about his/her own risk level than the insurance company does.

\medskip

In order to assess the benefit of an insurance for the household, we choose to study the scenario that would be the least efficient from the household point of view: we assume that the insurer is a monopoly. Among the literature on monopoly insurer models with adverse selection, we can mention the extension of the RS model of \citet{rothschild1976equilibrium} to a monopoly by \cite{stiglitz1977monopoly} and, more generally, the class of models with adverse selection (not in the field of insurance) developed by \citet{mirrlees1971exploration, spence1974competitive, guesnerie1984complete, salanie2005economics, laffont2009theory}, and more recently by \citet{guasoni2020sharing}, and their extensions to multi-products by, for example, \citet{armstrong1996multiproduct}, or in continuous-time by \citet{alasseur2020adverse}.
The study of this least efficient scenario of a monopoly model is also motivated by the application we have in mind, namely, a fuel-poverty\footnote{
According to the French law Grenelle 2, a \textit{`person in a fuel poverty situation is a person who has particular difficulties in his/her home in obtaining the necessary supply of energy to meet his/her basic needs because of the inadequacy of his/her resources or living conditions'\label{ft:def_fuel_poverty}}.} insurance.
Indeed, in this situation, the best-suited insurer is the customer's current energy supplier, who knows the customer better than other companies do.
Nevertheless, even if the insurer is a monopoly, the household has the choice of whether or not to subscribe an insurance policy among those offered by the monopoly.
As a result, we assume that the household will refuse any contracts if none provides more utility than expected without insurance, defined as the \textit{reservation utility}, which is thus endogenous and, in particular, depends on its risk.

\medskip

Inspired by the literature on standard contract theory with adverse selection, mainly the pioneering works of \citet{baron1982regulating, guesnerie1984complete, maskin1984monopoly} and the few applications to insurance problems such as \citet{stiglitz1977monopoly, landsberger1994monopoly, landsberger1996extraction, landsberger1999general}, or the more recent works of \citet{janssen2005dynamic, chade2012optimal, diasakos2018neutrally, bensalem2020prevention}, we model this situation as a principal-agent problem with adverse selection.
The principal (\textit{she} -- an insurance company, a supplier, etc.) can offer insurance to the agent (\textit{he} -- a household), which allows him to receive a specified amount of the staple good under consideration in the event of a loss of income. We assume that adverse selection concerns the agent's probability of income loss, defined as \textit{his type}. This assumption is classical in the literature, particularly in all extensions of \citet{rothschild1976equilibrium} and \citet{stiglitz1977monopoly}. The reservation utility we consider thus depends on the agent's type. The problem of endogenous reservation utility is studied in some adverse selection models, such as \citet{lewis1989countervailing, biglaiser1993principals, maggi1995countervailing, jullien2000participation,guasoni2020sharing}; and \citet{alasseur2020adverse} also discuss this issue for an application close to the one we have in mind. However, to the best of our knowledge, this problem is rarely addressed for a continuum of types in insurance models, as in our framework.

\medskip

One of the cornerstones of our approach is to combine a model of consumption (two goods, under budget constraint) with an insurance model under adverse selection. The advantage of this model lies in the fact that it allows, with a few adjustments, to deal with various situations in order to compare them. Indeed, although our study focuses on an insurance with benefits in kind, a very similar reasoning allows us to solve the case of a classic financial insurance. It is also possible to adapt the model, for example to take into account alternative options available for the household, as will be the case when considering the prepayment option. 

\medskip

Another particularity of our method lies in the choice of a two-period insurance model, as in \citet{schlesinger2014purchase, schlesinger2019whoops}.
Most of the literature on insurance with adverse selection focuses on only one period, where the agent pays for and receives the insurance at the same time, or on repeated versions of this scheme.
The closest application in which this type of two-period model has been developed is on the topic of self-prevention. In the works of \citet{eeckhoudt2012precautionary, wang2015precautionary, peter2017optimal} (model with savings) and \citet{menegatti2009optimal, courbage2012optimal} (without savings), the authors consider a two-period model to account for the delay between the prevention effort and its real benefit.
In our framework, such a setting seems necessary to capture the fact that the insured is not in a precarious situation when he subscribes to the insurance plan to be covered over the next period.
Considering a two-period model also allows us to compare, with and without insurance, the evolution of the agent's consumption when he suffers from a loss of income.

\medskip

The paper's contribution is threefold. First, we provide a rigorous formulation of the contracting problem between the households and the monopolist insurance company under adverse selection. This simple but not simplistic model enables the representation of both the consumption and insurance preferences of the household and, as mentioned above, can easily be adapted to take into account other types of insurance or option, such as prepayment. Second, we solve this principal-agent problem, using calculus of variation, and we obtain the most explicit results possible, despite the endogenous nature of the reservation utility. In particular, we show that the optimal contract can be characterised through the solution of a second-order non-linear ordinary differential equation. Even though we are working with a specific model, the \textit{modus operandi} we develop in this paper can readily be extended, for example to consider more general utility functions or other distributions of types in the population. We have therefore made particular efforts to detail the proof of our results in Appendix, believing that this could prove useful for further applications and generalisations. 
Finally, by means of numerical simulations, we show that an in-kind insurance, even provided by a monopoly, can be a tool to decrease the risk of fuel poverty among the population. Indeed, we prove that this mechanism ensures risky household a sufficient consumption of electricity, compare to a conventional financial assistance.
Even though our optimal insurance contract is characterised by a high price due to the monopoly situation, we show that this can be prevented by providing the household with other options, such as prepayment.

\medskip

The remainder of this paper is organised as follows. The principal-agent model is detailed in \Cref{sec:model}, and then solved under adverse selection in \Cref{sec:third_best}. 
In \Cref{sec:fuel_poverty_TB}, we apply our results to the context of fuel poverty, and compare the benefits of the in-kind insurance with both a financial insurance and a prepayment option.
\Cref{sec:conclusion} concludes by suggesting some policy recommendations, which might constitute relevant extensions of our model.

\section{The model}\label{sec:model}

\subsection{A principal-agent model with adverse selection}\label{ss:principal_agent}

We consider a two-period/two-good principal-agent model with adverse selection. The agent represents a household consuming the essential good considered and another representative good. More precisely, at each time $t \in \{0, 1\}$, the agent has an income $w_t$, which allows him to consume a quantity $e_t$ of the considered staple good and a quantity $y_t$ of another good, with respective unitary constant positive price $p_e$ and $p_y$. However, between the two periods, the agent is likely to suffer from a loss of income, which will put him in a precarious situation at time $t=1$: he will be obliged to reduce his consumption. However, if he does not consume a sufficient quantity of the staple good, this can lead to serious issues of which the agent is not necessarily aware. To prevent him from \textit{staple good poverty}, the risk-neutral principal, who may be the good producer or supplier, an insurance company, or even the government, offers insurance. This insurance ensures that the agent receives a specified quantity of the staple good, denoted $e_{\rm min}$, in the event of an income loss. At time $t=0$, the agent thus chooses if he wants to subscribe to the insurance plan, and if so, he pays the insurance premium $T$ associated with a contractible quantity $e_{\rm min}$. At time $t=1$, if he has purchased the insurance and if his income has decreased sufficiently, the agent receives the quantity $e_{\rm min}$ of the staple good.

\medskip

We define the random income of an \textit{agent of type $\varepsilon$} at time $t=1$ by $w_1 := \omega w_0$, where $w_0$ is the income at time $t=0$ and $\omega$ is a random variable, defined on the probability space $(\Omega, \Ac, \P^\varepsilon)$, where $\Omega$ is a subset in $\R_+$ and $\Ac$ is its natural $\sigma-$algebra. We assume that the insurance is activated when $\omega \leq \widetilde{\omega}$, where the income loss barrier $\widetilde{\omega}$ is set in an exogenous way. To obtain closed-form solutions, we make the following assumption:
\begin{assumption}\label{ass:income_loss}
The random variable $\omega$ takes two values, $\underline\omega$ with probability $\varepsilon$ and $\widebar{\omega}$ with probability $1-\varepsilon$,
where $\varepsilon \in [0,1]$ and $\widebar{\omega} > \widetilde{\omega} \geq \underline\omega > 0$.
\end{assumption}
We assume that the constants $\underline{\omega}$ and $\widebar{\omega}$ are common knowledge.
The inequality $\widebar{\omega} > \widetilde{\omega} \geq \underline\omega$ means that the insurance is only activated when $\omega = \underline \omega$. This model for the distribution of losses is traditional in insurance models based on the pioneering works of \citet{rothschild1976equilibrium} and \citet{stiglitz1977monopoly}. We consider that the agent is better informed than the principal about the risk of income loss he is facing, which depends on his work quality, his job insecurity, the relation he has with his supervisor, etc. The principal only has access to an overview of risks among the population, which leads to an adverse selection problem.
\begin{assumption}[Adverse selection]\label{ass:adverse_selection}
The principal cannot observe the type of an agent but knows the distribution of the types of her potential clients.
\end{assumption}

As usual in adverse selection problems, the principal has an incentive to offer a menu of contracts, \textit{i.e.}, various quantity $e_{\rm min}$ with the associated premia $T$.
The agent then chooses the contract that best suits him among all contracts offered by the principal, depending on his risk type. In our study, we look for the best continuous menu of contracts that the insurer can offer.

\subsection{Agent's problem}\label{ss:def_agent_pb}

In most insurance models, the agent's utility function is not specified. With the aim of obtaining the most explicit results possible, we choose here to represent the preferences of the agent over the goods' consumption at time $t$ by using a separable utility function based on logarithmic felicities:
\begin{align}\label{eq:utility}
    U(e_t,y_t) := \alpha \ln (e_t) + \ln(y_t), \; \text{ for } \; e_t, y_t > 0,
\end{align}
where $\alpha$ parametrises the \textit{longview} elasticity of substitution between the staple good and the composite good.
\medskip

Our model does not take into account the possibility for the agent to save between the two periods (contrary to \citet{schlesinger2014purchase}). This hypothesis may seem restrictive but is consistent with the literature on two-period models (in particular on prevention with \citet{menegatti2009optimal, courbage2012optimal}) and is justified in our framework in view of the particular households\footnote{Households that do not save but want to are widespread, as evidenced by the many mobile applications or services to help them.
The mobile application \textit{Birdycent} 
rounds up each payment made by the consumer to feed a piggy bank, with a zero interest rate, which is equivalent to losing money in relation to inflation. A second application, called \textit{Yeeld}, 
offers $4\%$ in cash back on Amazon instead of an interest rate. The bank \textit{Crédit Mutuel} proposes the service \textit{Budget +}, which is subject to a fee, to automatically save from a current account to a savings account.
In our opinion, this highlights that there is a need to encourage households to save money and that they are willing to pay for these types of services.} on which we want to focus in our study.
Indeed, one can assume that a household already used to saving has built up sufficient funds to pay its bills in the event of a loss of income. This household should thus not be concerned with the insurance we develop throughout this paper unless it has inadequate savings.
In fact, our insurance can precisely be interpreted as a form of incentive to save: it is a way for households that have no savings to obtain a quantity of staple good in case of an income loss.
The concept of insurance is very effective in this type of situation and allows risks to be shared among the population.
Moreover, this choice of model is also based on the willingness to keep a tractable model with (relatively) explicit solutions and to focus our study on the design of insurance contracts. Nevertheless, we will be led in the following to extend the model to consider a type of saving option, the prepayment of a quantity in $t=0$, delivered in $t=1$.

\subsubsection{Reservation utility}

Without insurance, the agent maximises, independently in each period $t$, the previously defined utility, under his budget constraint:
\begin{align}\label{eq:no_insurance}
    &\ V^\varnothing (w_t) := \max_{(e_t, y_t) \in \R_+^2} \; U(e_t,y_t), \; \text{ u.c. } \; e_t p_e + y_t p_y \leq w_t.
\end{align}
Given a discount factor $\beta \in [0,1]$, we define the intertemporal expected utility without insurance of an agent of type $\varepsilon$ as follows:
\begin{align}\label{eq:def_exp_utility_reservation}
    \textnormal{EU}^\varnothing (\varepsilon) := V^{\varnothing} (w_0) + \beta \E^{\P^\varepsilon} \big[ V^{\varnothing} (\omega w_0) \big].
\end{align}
In our framework, the agent is likely to accept the insurance contract only if it provides him with a level of utility at least equal to his utility without it. Therefore, the \textbf{reservation utility} of an agent of type $\varepsilon$ is defined\footnote{With the aim of simplifying the notation, we highlight only the dependency in type: the reservation utility is stated as a function of $\varepsilon$.} by \eqref{eq:def_exp_utility_reservation}.

\subsubsection{Expected utility with in-kind insurance}

Let us now fix an insurance contract $(e_{\rm min}, T)$.
If the agent decides to subscribe to this contract, we assume that the payment of the insurance premium $T$ only impacts his budget constraint at time $t=0$, and his maximum utility is thus naturally given by:
\begin{align}\label{eq:utility_t0}
    &\ V_0 (w_0, T) := V^\varnothing (w_0 - T).
\end{align}
As described in \Cref{ss:principal_agent}, the insurance we consider is an in-kind support: it ensures the agent a fixed non-negative amount $e_{\rm min} \geq 0$ of a determined staple good at time $t=1$ if he suffers a sufficient loss of income, \textit{i.e.}, if $\omega = \underline\omega$. Therefore, his maximisation problem is:
\begin{align}\label{eq:utility_t1}
    V_1 (\omega w_0, e_{\rm min}) := &\ \max_{(e_1, y_1) \in \R_+^2} \; U (e_1 + e_{\rm min} \mathds{1}_{\omega = \underline{\omega}}, y_1), \;
    \text{ u.c. } \; e_1 p_e + y_1 p_y \leq \omega w_0.
\end{align}
Similar to the case without insurance, we define the intertemporal expected utility of an agent of type $\varepsilon$ with an insurance contract $(e_{\rm min},T)$ by:
\begin{align}\label{eq:def_exp_utility}
    \textnormal{EU}^{\rm Q} (\varepsilon, e_{\rm min}, T) := V_0 (w_0, T) + \beta \E^{\P^\varepsilon} \big[ V_1 (\omega w_0, e_{\rm min}) \big].
\end{align}

\begin{remark}\label{rk:financial_insurance}
If one want to consider a financial insurance, offering an amount of money $r \geq 0$ at time $t=1$ if the agent suffers a sufficient loss of income, it suffices to set:
\begin{align*}
    &\ V_1 (\omega w_0, r) := V^\varnothing \big(\omega w_0 + r \mathds{1}_{\omega = \underline{\omega}} \big).
\end{align*}
\end{remark}

\subsection{Principal's problem}\label{ss:principal}

We assume that the principal is risk-neutral\footnote{The principal's risk-neutrality seems reasonable because shareholders of insurance companies generally have a diversified portfolio.} and wants to maximise her profit: she receives at time $0$ the earnings from the sale of insurance plans to agents of type $\varepsilon \in [0,1]$ who agree to subscribe but needs to provide them the quantity $e_{\rm min}$ they have chosen if they suffer from an income loss in the next period. We consider in this model that insurers are not in perfect competition, so that the price of insurance is not determined by the actuarial price.
Therefore, in this monopoly situation, the insurer can choose the range of $e_{\rm min}$ she wants to offer as well as the price associated with each quantity. We properly define the notion of admissible contracts and menu in our framework:
\begin{definition}\label{def:adm_contract_QI}
An admissible contract $\xi = (e_{\rm min}, T)$ is a quantity $e_{\rm min} \geq 0$ with an associated premium $T < w_0$.
An admissible menu is then defined as a continuum of admissible contracts $(e_{\rm min}, T)$, \textit{i.e.}, a continuum of non-negative quantities and a continuous price function $T$ defined for all quantities offered. Under \textnormal{ \Cref{ass:adverse_selection}}, the function price $T$ is required to be independent of the agent's type. We will denote the admissible set of menus in this case by $\Cc$.
\end{definition}

\textbf{In the first-best case}, the principal knows the type $\varepsilon$ of the agent and can thus offer him a particular contract. Since she has to pay with probability $\varepsilon$ the quantity $e_{\rm min}$ at the unitary price $p_e$, her optimisation problem is:
\begin{align}\label{eq:def_principal_pb_FB}
    \pi_{\varepsilon} := \sup_{(e_{\rm min}, T)} \big( T - \varepsilon p_e e_{\rm min} \big),
\end{align}
under the constraint that $(e_{\rm min}, T)$ is an admissible contract providing the agent with at least his reservation utility.

\medskip

\textbf{In the third-best case}, \textit{i.e.}, with adverse selection, we consider a menu of revealing contracts, in the sense that an agent of type $\varepsilon$ will subscribe to the insurance contract designed for him, \textit{i.e.}, $(e_{\rm min} (\varepsilon), T(\varepsilon))$. Assuming that the distribution of the type $\varepsilon$ in the population considered by the principal has a density function $f$, the principal's problem is defined as follows:
\begin{align}\label{eq:def_principal_pb_SB}
    \sup_{(e_{\rm min}, T) \in \Cc} \int_0^1 \big( T(\varepsilon) - \varepsilon p_e e_{\rm min} (\varepsilon) \big) f(\varepsilon) \drm \varepsilon,
\end{align}
under the participation constraint, and where, in this case, $e_{\rm min}$ and $T$ are appropriate functions of $\varepsilon$. To solve this case, we assume that the distribution of the type $\varepsilon$ in the population considered by the principal is uniform on $[0,1]$, \textit{i.e.} $f = \mathds{1}_{[0,1]}$.
This assumption is actually not necessary, as computations could easily be made for another distribution, but it allows for simplification of the principal's problem.
Moreover, this distribution models the problem of a principal who does not actually have data on the agent's income loss. This is the case in the application in question, where the insurer is an electricity supplier, who is not intended to have insight into the distribution of the risks of income loss of its customers. Moreover, even if a probability of one-half seems high for the population, the agents likely to subscribe to our insurance are people of a rather high risk type, and the probability in the population considered by the insurer is necessarily higher than in the global population. This is particularly true given that our study focuses on middle-class households without savings, which naturally have a higher probability of losing income.

\subsection{Application to fuel poverty}\label{ss:fuel_poverty}

We apply this model for a particular staple good, electricity, to develop insurance against fuel poverty. 
According to \cite{chaton2020simulation}, fuel poverty is essentially linked to a temporary loss of income.
It particularly affects low-income and vulnerable households that have a low propensity to save and that already spend a large part of their income on energy.
This situation can lead households to adopt risky behaviours, causing health problems and housing deterioration (see \citet{lacroix2015fuel}). For example, to keep heat inside their homes, some obstruct vents, thereby generating moisture and mould. Households in fuel poverty are often forced to make choices with harmful consequences for their health: choosing eating or heating, or giving up health care or social interactions.
This consequences are often neglected by households but are harmful, and moreover highly expensive for society. To avoid them, mechanisms are being developed to help vulnerable households. For example, in France, energy vouchers have been distributed by the state since 2018. This voucher can be used to pay not only for energy expenses such as electricity, gas, wood and fuel oil bills but also for energy renovation. In 2019, it targeted 5.8 million households with modest incomes. 

\medskip

The motivation of our work is to act on the prevention side by proposing a complementary tool to keep the number of households in fuel poverty from increasing. The idea is to develop an insurance policy that is activated if the household becomes energy constrained.
Two French electricity suppliers offer two slightly different insurance options: \textit{Assurénergie}, offered by Electricité de France (EDF), and \textit{Assurance Facture}, offered by ENGIE. These two monthly insurance plans offer a refund of part of the electricity bill in the event of job loss, sick leave, hospitalisation, disability or death. In the first insurance plan, the amount refunded depends on the contract chosen from the proposed menu, while the second insurance plan is a unique contract. A more general example is the \textit{Utilities Insurance} provided by the Canadian company Trans Global Insurance, which allows household to be covered for some basic utilities.
This insurer provides a menu of contracts in two ways: the household can choose both the level of coverage (3 possible) as well as the different basic utilities to be covered (power, heat, water, internet...). 

\medskip

In our framework, we consider an in-kind insurance, which can be interpreted as an instantaneous reimbursement of the quantity $e_{\rm min}$ on the invoice, and thus slightly differs from previously described insurances. Immediate repayment seems essential to help households in difficulty: indeed, the main problem is that these households are often unable to advance the required funds to pay their bills. The risk is that they anticipate their difficulties by drastically reducing their energy consumption for fear of the reimbursement delay. An in-kind support should help to prevent this bias.
However, an instantaneous refund can only be made by the consumer's supplier (to avoid significant transaction costs) and therefore the supplier seems to be in a monopoly situation for offering this type of insurance to the consumer. Therefore, the monopolistic framework under consideration makes even more sense in this situation, especially since the fuel provider of a household has more inside information than other fuel providers or traditional insurance companies. Our goal is to compute the optimal menu of contracts using contract theory with adverse selection to study the structure of the contracts obtained and to uncover what types of agents are likely to subscribe to the insurance plans. 

\section{Optimal insurance under adverse selection}\label{sec:third_best}

We first start by solving in \Cref{ss:agent_FB} the optimal consumption problem of an agent of type $\varepsilon$: given an insurance contract $(e_{\rm min}, T)$ and the utility function specified in \eqref{eq:utility}, we compute the agent's optimal consumption of both goods at each period. As a result, we can compute the maximum utility the agent can achieve for a given contract.
Comparing this utility with the reservation utility, we can determine the maximum price the agent is willing to pay for insurance. This first part allows us to properly define in our context the \textit{participation constraint} mentioned in the definition of the principal's problem in \Cref{ss:principal}. In particular, with this in mind, we can solve the the first-best case problem, \textit{i.e.}, without adverse selection. This intermediate result, postponed to \Cref{sec:first_best}, gives some insights to solve the problem in the interesting case, \textit{i.e.}, under adverse selection.

\medskip

We then focus on finding the optimal menu of insurance contracts in the presence of adverse selection. The classical scheme for doing so in this case is to use the \textit{revelation principle}: the principal has to design a menu of contracts indexed by $\varepsilon$, such that the agent of type $\varepsilon$ chooses the contract designed for him. 
This revelation principle detailed in \Cref{ss:RP_QI} allows us to write the insurance premium as a function of the insured quantity and the type to within a constant $c_q$. In fact, \Cref{ss:PC} explains that the value of this constant is related to the participation constraint of the agents: the principal can choose the constant depending on the type of agents she wants to select. \Cref{ss:principal_SB} is dedicated to solving principal's problem. The numerical application to fuel poverty is postponed to the next section.

\subsection{Solving the agent's problem}\label{ss:agent_FB}

We first solve the consumption problem of an agent who has not subscribed to an insurance contract. With this in mind, let us define the following constant:
\begin{align}\label{eq:constant_useless}
    C_{\alpha,p_e,p_y} &:= \alpha \ln (\alpha) - (1+\alpha) \ln ( 1+ \alpha) - \alpha \ln (p_e) - \ln(p_y).
\end{align}
Since our framework does not allow the agent to transfer income from one period to another, the agent maximises his utility in each period independently by solving \eqref{eq:no_insurance}, which leads to the following result.
\begin{lemma}[Without insurance]\label{lem:no_insurance}
The optimal consumptions of each good at time $t \in \{0,1\}$ of an agent with income $w_t$ are given by:
\begin{align*}
        y_t^{\varnothing} := \dfrac{1}{1+\alpha} \dfrac{w_t}{p_y} \; \text{ and } \; e_t^{\varnothing} := \dfrac{\alpha}{1+\alpha} \dfrac{w_t}{p_e},
    \end{align*}
and induce the maximum utility $V^{\varnothing} (w_t) = (1+\alpha) \ln (w_t) + C_{\alpha,p_e,p_y}$.
\end{lemma}

Then, by a simple computation of the expected utility defined by \eqref{eq:def_exp_utility_reservation}, we can explicitly write the reservation utility:
\begin{proposition}\label{prop:reservation_utility}
Under \textnormal{\Cref{ass:income_loss}}, the expected utility without insurance of an agent of type $\varepsilon$ is given by
\begin{align}\label{eq:reservation_utility}
\textnormal{EU}^\varnothing (\varepsilon) 
&= (1+\alpha) \ln \big( \underline\omega^{\beta \varepsilon} \widebar{\omega}^{\beta (1-\varepsilon) } w_0^2 \big) + (1+ \beta) C_{\alpha,p_e,p_y}.
\end{align}
\end{proposition}

Usually, in insurance models, the reservation utility is taken to be independent of the agent's type. In our framework, we reason that an agent does not subscribe to an insurance contract if his utility without is higher.
Therefore, the reservation utility we consider is endogenous and depends on the probability $\varepsilon$. In particular, $\textnormal{EU}^\varnothing$ is a decreasing function of $\varepsilon$. This problem is addressed in some adverse selection models, for example, \citet{lewis1989countervailing, biglaiser1993principals, maggi1995countervailing, jullien2000participation, alasseur2020adverse}, but is rarely considered in insurance problems. As \citet{laffont2009theory} explain, in this case, determining which participation and incentive constraints are binding becomes a more difficult task. Nevertheless, \Cref{prop:PC_for_RevContracts} establishes that only agents of a sufficiently risky type are selected by the principal, and this feature only appears in models with countervailing incentives. The principal excludes the types with lower risk, those with a low probability of losing their income, because the price they are willing to pay is very low, while the agents of a higher risk type are more profitable, since they are easily satisfied and are willing to pay much more.

\medskip

Similarly, we can solve the consumption problem of an agent who subscribes to a given admissible contract (see \Cref{lem:max_utility_t0,lem:max_utility_t1}).
Without loss of generality, we can assume that any admissible contract $(e_{\rm min}, T)$, in the sense of \Cref{def:adm_contract_QI}, is of the following form:
\begin{align}\label{eq:normalised_contract}
    e_{\rm min} := q \alpha \underline \omega w_0 /p_e, \; \text{ for } \; q \in \R_+ \; \text{ and } \; T := t_0 w_0, \; \text{ for } \; t_0 \in [0,1).
\end{align}
The pair $(q, t_0)$ is referred to as an \textbf{admissible normalised contract}. We then denote by $\widebar U$ the following function:
\begin{align}\label{eq:overline_U}
    \widebar U(q) :=
    \left\{
    \arraycolsep=1.4pt\def\arraystretch{1.3}
        \begin{array}{ll}
        (1+\alpha) \ln (1 + q \alpha) & \mbox{ if } q < 1, \\
        \alpha \ln (q) + (1+\alpha) \ln ( 1+ \alpha) & \mbox{ if } q \geq 1,
        \end{array}
    \right.
    \text{ for } \; q \in \R_+.
\end{align}
The preliminary results in \Cref{app:with_insurance} allow us to provide an explicit form in the following proposition for the expected utility defined by \eqref{eq:def_exp_utility}.

\begin{proposition}\label{prop:exp_utility}
Given an admissible normalised contract $(q, t_0)$ and under \textnormal{\Cref{ass:income_loss}}, the expected utility of an insured agent of type $\varepsilon$ is given by:
\begin{align}\label{eq:exp_utility}
    &\ \textnormal{EU}^{\rm Q} (\varepsilon, q, t_0) = \textnormal{EU}^\varnothing (\varepsilon) + (1+\alpha) \ln (1 - t_0) + \beta \varepsilon \widebar U(q).
\end{align}
\end{proposition}

\begin{remark}\label{rk:q=1}
The separation of cases between $q$ less or greater than $1$ is related to the fact that the agent is not allowed to resell\footnote{If the agent could resell part of the quantity, the insurance would be strictly equivalent to income insurance, which is why we ignore any reselling possibility.} part of the insured quantity $($the quantity of the good provided under the insurance plan, i.e., $e_{\rm min})$. If $q<1$, this quantity is not sufficient from the agent's point of view. He therefore supplements this quantity by purchasing additional energy at $t=1$.
In contrast, if $q \geq 1$, the agent consumes only the corresponding amount $e_{\rm min}$ of the staple good, and his optimal complementary consumption is equal to zero. However, in this case, it could have been better from his point of view to resell part of the insured quantity. Nevertheless, it is precisely the purpose of this paper to ensure that the household consumes more of this particular good to avoid the consequences induced by a decrease in consumption, of which the household is not aware.
\end{remark}

It remains to be determined when the agent of type $\varepsilon$ will subscribe the insurance plan, \textit{i.e.}, when his expected utility with insurance is greater than his reservation utility. With this in mind, by computing the difference between \eqref{eq:def_exp_utility} and \eqref{eq:def_exp_utility_reservation}, we can state the following proposition.
\begin{proposition}[Participation constraint]\label{prop:PC}
An admissible normalised contract $(q,t_0)$ satisfies the participation constraint of the agent of type $\varepsilon$ if and only if $t_0 \leq t_{max} (\varepsilon, q)$, where $t_{max}$ is defined for any $(\varepsilon, q) \in [0,1] \times \R_+$ by:
\begin{align}\label{eq:max_price_QI}
    t_{max} (\varepsilon, q) := 1 -
    \left\{
    \arraycolsep=1.4pt\def\arraystretch{1.2}
        \begin{array}{ll}
        (1 + q \alpha)^{- \beta \varepsilon} & \mbox{ if } q < 1, \\
        q^{- \beta \varepsilon \frac{\alpha}{1+\alpha}} (1+ \alpha)^{- \beta \varepsilon} & \mbox{ if } q \geq 1.
        \end{array}
    \right.
\end{align}
\end{proposition}
Therefore, $w_0 t_{max} (\varepsilon, q)$ is the maximum price the agent of type $\varepsilon$ is willing to pay for a quantity $e_{\rm min} = \alpha q \underline\omega w_0/p_e$. In other words, if the premium $T$ associated with a quantity $e_{\rm min} = \alpha q \underline\omega w_0/p_e$ is below $w_0 t_{max} (\varepsilon, q)$, the agent of type $\varepsilon$ is willing to purchase the insurance contract. We say, in this case, that the admissible contract $(e_{\rm min}, T)$ satisfies the participation constraint for $\varepsilon$ types. On the contrary, if the premium $T$ is above $w_0 t_{max} (\varepsilon, q)$, the agent will not purchase the insurance, meaning that the participation constraint is not satisfied.

\medskip

Through \Cref{eq:max_price_QI}, the maximum price seems to be independent of the price $p_e$. In reality, this is only due to the fact that $q$ represents a quantity normalised (among other parameters) by price. More precisely, the maximum price $w_0 t_{max} (\varepsilon, q)$ can be reformulated as a function of $e_{\rm min}$:
\begin{align}\label{eq:max_price_QI_emin}
    T_{max} (\varepsilon, e_{\rm min}) := w_0 -
    \left\{
    \arraycolsep=1.4pt\def\arraystretch{2}
        \begin{array}{ll}
        w_0 \bigg(1 + \dfrac{p_e e_{\rm min}}{\underline\omega w_0} \bigg)^{- \beta \varepsilon} & \mbox{ if }  e_{\rm min} < \alpha \underline\omega w_0/p_e, \\
        w_0 \bigg(\dfrac{p_e e_{\rm min}}{\alpha \underline\omega w_0} \bigg)^{- \beta \varepsilon \frac{\alpha}{1+\alpha}} (1+ \alpha)^{- \beta \varepsilon} & \mbox{ if } e_{\rm min} \geq \alpha \underline\omega w_0/p_e.
        \end{array}
    \right.
\end{align}
Thanks to this formula, it is straightforward to make the following comparative statics.
\begin{corollary}\label{cor:comp_stat_tmax}
The function $T_{max} (\varepsilon, e_{\rm min})$ defined by \eqref{eq:max_price_QI_emin} is increasing with respect to $e_{\rm min}$, $w_0$, $p_e$, $\varepsilon$, and $\beta$, but decreasing with respect to $\underline \omega$.
Moreover, when $e_{\rm min} < \alpha \underline\omega w_0/p_e$, $T_{max}$ is independent of $\alpha$.
\end{corollary}

The results presented in the previous corollary (see \Cref{app:tech_proof_CS} for the proof) are, for the most part, very intuitive. First, the agent is naturally willing to pay more for a higher insured quantity. 
In addition, the higher the agent's initial income $w_0$, the more he is willing to pay a high price. Conversely, the higher his income $w_1$ in $t=1$ ($\underline \omega$ close to $1$), the less he is willing to pay a high price. Indeed, if the agent is not likely to lose significant income between $t=0$ and $t=1$, this kind of insurance is of limited value to him. 
Moreover, if the staple good is expensive ($p_e$ high), the agent is willing to pay more for the insurance. Similarly, the more the agent values the future ($\beta$ large), the more he is ready to spend a large part of his income in $t=0$ by buying the insurance in order to benefit from it in $t=1$.

\medskip

All the aforementioned points seem to be in line with reality and therefore validate the choice of our model.
Nevertheless, the main feature to observe is that the maximum price increases with the type $\varepsilon$ of the agents, \textit{i.e.}, the probability of income loss. Therefore, the riskier the agent is, the more willing he is to pay a high price for the same insured quantity. Finally, the fact that for a small insured quantity ($e_{\rm min} < \alpha \underline\omega w_0/p_e$), the maximum price is independent of $\alpha$, is related to \Cref{rk:q=1}. Indeed, in this case, the insurance acts as an income insurance, since the agent would have consumed the amount $e_{\rm min}$ anyway, even without insurance. The maximum price is therefore independent of his preference between the essential good and the composite good. On the other hand, when the insured quantity is large, the agent would prefer to resell the surplus. This would allow him an additional income that he could then optimally distribute between the essential good and the composite good, according to his preference represented by the parameter $\alpha$. Following the reasoning developed in the proof of \Cref{cor:comp_stat_tmax}, we can state that the maximum price is increasing with respect to $\alpha$ for $\alpha$ sufficiently small and then decreasing. This result is also expected as $\alpha$ small means that the agent has a strong preference for the essential good, and is thus very inclined to buy insurance. On the contrary, when $\alpha$ gets close to $1$, the agent becomes indifferent between the two goods, and is therefore less inclined to buy the insurance.

\subsection{Revelation principle}\label{ss:RP_QI}

Traditionally, in adverse selection models (see \citet{salanie2005economics} for the general theory on adverse selection), the contract offered by the principal has to satisfy the \textit{incentive compatibility $(IC)$ constraint}: the contract has to be such that an agent of type $\varepsilon$ would subscribe the contract corresponding to him, and thus reveal his type $\varepsilon$, which is previously unknown to the principal. Indeed, the well-known \textit{revelation principle} implies that we can restrict the study to incentive compatible mechanisms. More precisely, the revelation principle stated for example in \citet{salanie2005economics} can be adapted to our framework as follows:
\textit{If the optimal quantity $e_{\rm min}$ chosen by an agent of type $\varepsilon$ can be implemented through some mechanism, then it can also be implemented through a direct and truthful mechanism where the agent reveals his risk $\varepsilon$.}

\medskip

First, we can show that the Spence-Mirrlees condition, also called the \textit{constant sign assumption} in \citet{guesnerie1984complete}, as defined in \citet{laffont2009theory}, is automatically satisfied in our framework (see \Cref{lem:sm_condition}). This property makes the incentive problem well behaved in the sense that only local incentive constraints need to be considered. This condition was introduced by \citet{spence1973job} in his theory of signaling on the labour market and, similarly, by \citet{mirrlees1971exploration} in his theory of optimal income taxation, as the \textit{single-crossing assumption}: the condition implies that the indifference curves of two different types of agents can only cross once. This condition also has economic content, which implies in our framework that agents with a higher probability of income loss are willing to pay more for a given increase in $e_{\rm min}$ than agents of a less risky type. This condition ensures that it is possible to separate the agents of a high risk type from those of a low risk type by offering them better coverage in exchange for a higher premium.

\medskip

To find the risk-revealing contracts, we define, for an admissible menu of contracts $(e_{\rm min}, T)$, an associated pair $(q, t_0)$ of functions of $\varepsilon \in [0,1]$. 

\begin{definition}\label{def:regularity}
A mechanism $(q,t_0)$ is said to be admissible if
\begin{enumerate}[label=$(\roman*)$]
\item $q$ and $t_0$ are continuous functions on $[0,1]$, taking values in an interval contained in $\R_+$ and $[0, 1)$;
\item $q$ and $t_0$ have continuous first and second derivatives on $(0,1)$, except at $\varepsilon_1 := \min \{ \epsilon \in [0,1], \; s.t. \; q(\epsilon) = 1 \}$\footnote{With the convention that $\varepsilon_1 = 0$ if $q (\epsilon) \geq 1$ for all $\epsilon \in [0,1]$.}.
\end{enumerate}
More precisely, an admissible menu of contracts $(e_{\rm min}, T)$ is associated with an admissible mechanism $(q, t_0)$ if, for all quantities $e_{\rm min}$ available at price $T$, there is an $\varepsilon \in [0,1]$ such that $e_{\rm min} = \alpha q(\varepsilon) \underline\omega w_0/p_e$ and $T = w_0 t_0(\varepsilon)$.
\end{definition}

The first point of the previous definition is entirely based on assumptions made about an admissible menu of contracts in \Cref{def:adm_contract_QI}. The second point on the regularity of $q$ and $t_0$ is more a technical assumption made to simplify the reasoning: it allows us to use the first- and second-order conditions to define an incentive compatible contract. Unfortunately, the case separation between $q < 1$ and $q \geq 1$ subsequently implies a loss of $\Cc^1$ continuity of the quantity and price at this point.

\medskip

We thus limit our study to mechanisms that are smooth enough in the sense of \Cref{def:regularity} $(ii)$. According to the reasoning of \cite{guesnerie1984complete}, our results could be easily extended to piecewise continuously differentiable mechanisms of class $\Cc^1$, and some could even be generalised to all mechanisms. Nevertheless, significant additional difficulties can be avoided with this smoothness assumption. Moreover, one may note that the optimal mechanism in the first-best case is smooth in the sense of \Cref{def:regularity} $(ii)$ and that it therefore makes sense to restrict our study in this way.

\medskip

The IC constraint says that the utility of an agent of type $\varepsilon \in [0,1]$ has to be maximal for the choice of the contract $(q(\varepsilon), t_0(\varepsilon))$, \textit{i.e.},
\begin{align*}
    \textnormal{EU}^{\rm Q} \big(\varepsilon, q(\varepsilon), t_0(\varepsilon) \big)
     = \max_{\varepsilon' \in [0,1]} \textnormal{EU}^{\rm Q} \big( \varepsilon, q(\varepsilon'), t_0(\varepsilon') \big).
\end{align*}
In other words, if a menu of contracts satisfies the IC constraint, then the agent has an interest in revealing his type by choosing the contract made for him. We denote by $\Cc^{\rm Q}$ the set of admissible mechanisms satisfying this constraint. With the aim of lightening the equations, we denote throughout this section:
\begin{align}\label{eq:Q_0}
Q_0 ( \varepsilon )
:= 
\left\{
\arraycolsep=1.6pt\def\arraystretch{1.8}
    \begin{array}{ll}
    \displaystyle\int_{0}^\varepsilon \ln \big( 1 + \alpha q(\epsilon) \big) \drm \epsilon
    & \mbox{ if } \varepsilon \in [0, \varepsilon_1) \\
    \displaystyle\int_{0}^{\varepsilon_1} \ln \big( 1 + \alpha q(\epsilon) \big) \drm \epsilon
    + \dfrac{\alpha}{1+\alpha} \displaystyle \int_{\varepsilon_1}^\varepsilon \ln \big( q(\epsilon) \big) \drm \epsilon
    & \mbox{ if } \varepsilon \in [\varepsilon_1, 1].
    \end{array}
\right.
\end{align}
\begin{theorem}\label{thm:IC_SB}
An admissible mechanism $(q,t_0)$ satisfies the IC constraint for all $\varepsilon \in [0,1]$ if and only if the function $q$ is non-decreasing on $[0,1]$ and there exists $c_q \geq 0$ such that the price $t_0$ satisfies for all $\varepsilon \in [0,1]$,
\begin{align}\label{eq:TQ_emin_eps}
t_0(\varepsilon) = 1 - c_q  \erm^{\beta Q_0(\varepsilon)} \times 
\left\{
\arraycolsep=1.6pt\def\arraystretch{1.5}
    \begin{array}{ll}
    \big( 1 + \alpha q (\varepsilon) \big)^{- \beta \varepsilon}
    & \mbox{ if } \varepsilon \in [0, \varepsilon_1), \\
    \big( 1 + \alpha \big)^{- \beta \varepsilon_1} \big( q (\varepsilon) \big)^{- \beta \varepsilon \alpha/(1+\alpha)}
    & \mbox{ if } \varepsilon \in [\varepsilon_1, 1].
    \end{array}
\right.
\end{align}
\end{theorem}

The previous proposition provides a characterisation of an admissible mechanism $(q, t_0)$ satisfying the IC constraint for all types of agents; its proof is postponed to \Cref{app:tech_proof_revealing}.
Nevertheless, the concrete menu of contracts proposed by the principal must be composed of quantities $e_{\rm min}$ and a price $T$ associated with each quantity, regardless of the type of agent, as specified in \Cref{def:adm_contract_QI}. The form of the practical menu of contracts associated with an admissible mechanism is a consequence of the previous theorem and is given by \Cref{cor:contract_emin}. We can summarise this result by saying that considering a sufficiently smooth admissible menu of revealing contracts $(e_{\rm min}, T)$ is equivalent to considering an admissible mechanism $(q,t_0)$, where $q$ is non-decreasing and the price $t_0$ is given by \eqref{eq:TQ_emin_eps}. It is now necessary to establish conditions implying that such a mechanism satisfies the agent's participation constraint.

\subsection{Adding the participation constraint}\label{ss:PC}

Recall that an agent of type $\varepsilon \in [0,1]$ accepts the contract if his utility from it is larger than his reservation utility, defined in our framework as his utility without insurance. To establish a precise result, we define the function $\underline c$ by:
\begin{align}\label{eq:underline_c}
\underline c (\varepsilon)
:= 
\left\{
\arraycolsep=1.5pt\def\arraystretch{1.2}
    \begin{array}{ll}
    \erm^{-\beta Q_0(\varepsilon)}
    & \mbox{ if } \varepsilon < \varepsilon_1, \\
    (1+ \alpha)^{- \beta (\varepsilon - \varepsilon_1)} 
    \erm^{-\beta Q_0(\varepsilon)}
    & \mbox{ if } \varepsilon \geq \varepsilon_1,
    \end{array}
\right. \text{ for all } \; \varepsilon \in [0,1].
\end{align}

The following proposition states that by controlling the constant $c_q$ in the insurance premium $t_0$ given by \eqref{eq:TQ_emin_eps}, the principal can choose whether to select agents of lower risk types. As a result, only the agents of a sufficiently high risk type are selected by the principal. Indeed, agents with a high probability of losing their income are easily satisfied and willing to pay much more than those of a less risky type. This result is entirely implied by the fact that the reservation utility of an agent depends on his type and only occurs in principal-agent problems with countervailing incentives. Additional information including the proof of the proposition can be found in \Cref{app:tech_proof_participation}. In particular, \Cref{rk:constant_reservation} shows that if a constant reservation utility had been chosen, the selected agents would have been those of the less risky type.
\begin{proposition}\label{prop:PC_for_RevContracts}
If the mechanism $(q,t_0)$ is admissible and incentive compatible, an agent of type $\varepsilon \in [0, 1]$ subscribes to the insurance plan if and only if $c_q \geq \underline c (\varepsilon)$. Moreover, by defining $\underline \varepsilon := \min \{ \varepsilon \in [0,1], \; \text{s.t.} \; c_q = \underline c(\varepsilon)\}$, the participation constraint is satisfied only for agents of type $\varepsilon \in [\underline \varepsilon, 1]$.
\end{proposition}

Now that the subset of risk-revealing contracts satisfying the agents' participation constraint is well defined, we can study the principal's problem.

\subsection{The optimal menu of contracts}\label{ss:principal_SB}

In the third-best case, the principal's goal is to find an optimal admissible menu of contracts $(e_{\rm min}, T)$ to maximise her profit, as defined by \eqref{eq:def_principal_pb_SB}, without knowing the agent's type. In fact, instead of maximising the principal's utility over all possible contracts, we restrict the study to the menus of contracts associated with an admissible mechanism $(q,t_0)$ in the sense of \Cref{def:regularity}. Thus, in line with the revelation principle, it is sufficient to only consider admissible mechanisms that are risk revealing. Recalling that an agent subscribes to a contract if and only if it satisfies his participation constraint, the principal's problem becomes:
\begin{align}\label{eq:principal_pb_simplify}
    \sup_{(q,t_0) \in \Cc^{\rm Q}} \int_{\varepsilon \in \Xi(q, t_0)} \pi(\varepsilon) \drm \varepsilon, \; \text{ for } \; \pi (\varepsilon) := w_0 t_0(\varepsilon) - \varepsilon \alpha q( \varepsilon) \underline\omega w_0,
\end{align}
where $\Xi(q,t_0)$ denotes the set of $\varepsilon \in [0,1]$ such that the participation constraint $t_0(\varepsilon) \leq t_{max} (\varepsilon, q (\varepsilon))$ is satisfied, where $t_{max}$ is defined by \eqref{eq:max_price_QI}.

\medskip

From \Cref{thm:IC_SB}, we know that an admissible mechanism $(q,t_0)$ satisfies the IC constraint if and only if $q$ is increasing and the price $t_0$ is given by \eqref{eq:TQ_emin_eps}. Moreover, we know from \Cref{prop:PC_for_RevContracts} that the participation constraint is satisfied only for agents of type $\varepsilon \in [\underline \varepsilon, 1]$ if and only if $c_q = \underline c(\underline \varepsilon)$, where $\underline c$ is defined by \eqref{eq:underline_c}. For $\underline \varepsilon \in [0,1]$, we thus denote by $\Cc^{\rm Q}(\underline \varepsilon)$ the set of admissible and revealing mechanisms such that the participation constraint is satisfied for all $\varepsilon \in [\underline \varepsilon, 1]$ only. In line with the previous reasoning, the principal's problem is equivalent to:
\begin{align*}
    \sup_{\underline \varepsilon \in [0,1]} \Pi (\underline \varepsilon), \; \text{ where } \; \Pi (\underline \varepsilon) := &\ \sup_{(q,t_0)\in \Cc^{\rm Q}(\underline \varepsilon)} \int_{\underline \varepsilon}^1 \pi(\varepsilon) \drm \varepsilon.
\end{align*}

To solve the principal's problem, we first fix $\underline \varepsilon \in [0,1]$. We denote by $\Qc(\underline \varepsilon)$ the space of functions $Q$ defined on $[\underline \varepsilon, 1]$, continuous and piecewise continuously differentiable of class $\Cc^3$, satisfying:
\begin{enumerate}[label=$(\roman*)$]
\item $Q$ is continuous on $[\underline \varepsilon, 1]$ and such that $Q(\underline \varepsilon) = 0$;
\item $Q'$ is positive and continuous except at $\varepsilon_1$, where $Q'(\varepsilon_1^{-}) = \ln (1+\alpha)$ and $Q'(\varepsilon_1^{+}) = 0$;
\item $Q''$ is positive and continuous except at $\varepsilon_1$.
\end{enumerate}
We consider the following second-order non-linear ordinary differential equation (ODE):
\begin{align}\label{eq:ODE}
    \dfrac{\beta}{\underline\omega} \big( \beta  \varepsilon^2 Q''(\varepsilon) - 2 \big)
    \erm^{\beta (Q(\varepsilon) - \varepsilon Q'(\varepsilon))}
    + G(\varepsilon, Q) = 0,
\end{align}
with initial conditions $Q(\underline \varepsilon) = 0$ and $Q'(\underline \varepsilon) = \underline q$, for $\underline q \in \R_+$, and where the function $G$ is defined for any $(\varepsilon, Q) \in [\underline \varepsilon, 1] \times \Qc(\underline \varepsilon)$ by:
\begin{align*}
G(\varepsilon, Q) :=
\left\{
\arraycolsep=1.4pt\def\arraystretch{1.6}
\begin{array}{ll} 
    \big( 1 + \varepsilon Q''( \varepsilon ) \big) \erm^{Q'( \varepsilon )}, 
    & \mbox{ for } \varepsilon \in [\underline \varepsilon,  \varepsilon_1 \vee \underline \varepsilon), \\
    \big( 1 + \alpha \big)^{\beta (\varepsilon_1 \vee \underline \varepsilon)+1}
    \bigg( 1 + \varepsilon \dfrac{1+\alpha}{\alpha} Q''(\varepsilon) \bigg)
    \erm^{\frac{1+\alpha}{\alpha} Q' ( \varepsilon )}, 
    & \mbox{ for } \varepsilon \in [\varepsilon_1 \vee \underline \varepsilon, 1].
\end{array}
\right.
\end{align*}
This ODE is at the heart of the resolution of the principal's problem, since it characterises the optimal admissible mechanism, for $\underline \varepsilon$ and $\underline q$ fixed.

\begin{theorem}\label{thm:sol_principal_ODE}
Given $\underline \varepsilon \in [0,1]$ and $\underline q \in \R_+$, if there exists $Q \in \Qc(\underline \varepsilon)$ solution to the \textnormal{ODE \eqref{eq:ODE}}, then the optimal admissible mechanism $(q^\star,t_0^\star)$ for the principal is given by:
\begin{align*}
\arraycolsep=1.5pt\def\arraystretch{1.8}
    \begin{array}{ll}
    \bigg(\dfrac{1}{\alpha}\big( \erm^{Q' ( \varepsilon )} -1 \big),  1 - \erm^{\beta (Q(\varepsilon) - \varepsilon Q'(\varepsilon))} \bigg)
    & \mbox{ for } \varepsilon \in [\underline \varepsilon,  \varepsilon_1 \vee \underline \varepsilon) \\
    \bigg( \erm^{\frac{1+\alpha}{\alpha} Q' ( \varepsilon )}, 1 - \big( 1 + \alpha \big)^{- \beta (\varepsilon_1 \vee \underline \varepsilon)} 
    \erm^{\beta (Q(\varepsilon)- \varepsilon  Q' ( \varepsilon ))} \bigg)
    & \mbox{ for } \varepsilon \in [\varepsilon_1 \vee \underline \varepsilon, 1].
    \end{array}
\end{align*}
\end{theorem}

\begin{remark}\label{rk:ODE_necessary}
\textnormal{\Cref{thm:sol_principal_ODE}} only gives a sufficient condition for the principal's optimisation problem. In fact, it would be possible to obtain a necessary condition. Nevertheless, in the numerical example we are interested in $($detailed in the following subsection$)$, as the solution of \textnormal{ODE \eqref{eq:ODE}} naturally satisfies the constraint of being in $\Qc(\underline \varepsilon)$, we decide to simplify the result by presenting it in this way. For more details, the reader is referred to \textnormal{\Cref{rk:ODE_constraint}}.
\end{remark}

For the sake of clarity, the proof of the theorem is reported in \Cref{app:tech_proof_principal}. The main point to observe is that, contrary to the first-best case, the optimal insured quantity, $e_{\rm min}^\star (\varepsilon) = \alpha q^\star(\varepsilon) \underline \omega w_0/p_e$, is now increasing with respect to $\varepsilon$. Since the premium is also increasing with respect to $\varepsilon$, no agent should have an interest in hiding his risk anymore. 


\medskip

However, the ODE \eqref{eq:ODE} cannot be solved other than numerically. Therefore, to solve the principal's problem, one must first fix $\underline \varepsilon \in [0,1]$ and an arbitrary initial value $\underline q \in \R_+$ for $Q'(\underline \varepsilon)$. Then, the solution of the previous ODE can be computed. With this solution, one can compute the principal's profit in this case, using \Cref{cor:pb_principal_SB}. This profit can then be maximised by choosing an optimal initial condition $\underline q$ and an optimal $\underline \varepsilon \in [0,1]$. For the numerical results, readers are referred to the next subsection, which discusses the application of this model to a particular framework: fuel poverty.

\section{Efficiency of in-kind insurance to fight fuel poverty}\label{sec:fuel_poverty_TB}

With the motivation described in \Cref{ss:fuel_poverty}, we apply the previous results in the context of an insurance against fuel poverty in France. 

\subsection{Numerical approach}\label{ss:fuel_poverty_param}

Many French households that belong to the poorest 30\% are in fuel poverty. Indeed,
in 2018, in mainland France, 78\% of households in the first income decile were in fuel poverty. This percentage falls to 54\% (26\%) for the 2nd (3rd) income decile.
These households are already struggling to have the necessary energy to avoid fuel poverty and thus do not have the financial means to subscribe to fuel poverty insurance plans.

\medskip

To perform our numerical simulations, we therefore consider a middle-class household with an annual disposable income, after taxes and social benefits, of $w_0=35,000$ \euro{}.
We assume that this household lives in an all-electric house (electric heating and hot water), with an annual electricity consumption\footnote{Note that this consumption is approximately the average electricity consumption of a French household, which is equal to $14,527$ kWh, according to \cite{belaid2016understanding}.} of $e_0^{\varnothing}=14,403$ kWh. Since 2018, in France, the average price $p_e$ per kilowatt hour was $0.18$ \euro{}, and the share of household income spent on household energy expenditure was $7.41$\%. We deduce from the expression of $e_0^{\varnothing}$ the value $\alpha = 8\%$.
Moreover, we set $\underline\omega = 0.4$: the household thus has a probability $\varepsilon$ of having an annual disposable income equal to $0.4 \times 35,000$ \euro{} $ = 14,000$ \euro{} at time $t = 1$. To simplify, we assume $\beta =1$.

\medskip

From the recursive scheme explained in the previous subsection, we obtain the optimal $\underline \varepsilon^\star \approx 0.63$ with $Q'(\underline \varepsilon^\star) \approx 0$, and thus, agents of type $\varepsilon < 0.63$ are not insured. Solving ODE \eqref{eq:ODE} for these parameters, we obtain the optimal function $Q \in \Qc(\underline \varepsilon^\star, 0)$, and in particular $\varepsilon_1^\star \simeq 0.66$. Thanks to \Cref{thm:sol_principal_ODE}, we can compute the optimal revealing mechanism $(q^\star,t_0^\star)$, and thus the optimal insured quantity for an agent of type $\varepsilon \in [\underline \varepsilon^\star, 1]$, which is $e_{\rm min}^\star (\varepsilon) = \alpha q^\star(\varepsilon) \underline \omega w_0/p_e$, and its corresponding price $T^\star(\varepsilon) = w_0 t_0^\star(\varepsilon)$. 

\subsection{Optimal quantity insurance}

Optimal quantities and prices are represented by blue curves in \Cref{fig:quantity_SB}. The left graph shows that the insured quantity in the third-best case (blue curve) is smaller than that in the first-best case (green curve). Only agents of type $\varepsilon = 1$ are insured for the same quantity as they would be in the first-best case, \textit{i.e.}, $e^\star_{\rm min}(1) \approx 12,665$ kWh/year. In the right graph, the green curve represents the maximum price an agent of type $\varepsilon$ is willing to pay for the quantity $e^\star_{\rm min} (\varepsilon)$, given by $T_{max} (\varepsilon) := w_0 t_{max} (\varepsilon, q^\star(\varepsilon))$. We thus observe that an agent of type $\underline \varepsilon^\star$ pays his maximum price, and his information rent is thus reduced to zero. In contrast, an agent of type $\varepsilon > \underline \varepsilon^\star$ obtains an information rent. In particular, the information rent is increasing with $\varepsilon$.

\medskip

However, one should already note a limit of the model: the comparison between blue and orange curves in the right graph shows that the agents are willing to pay more than the future price of the quantity, \textit{i.e.}, $p_e e_{\rm min} (\varepsilon)$.
This fact highlights a significant inconsistency of these agents, who are willing to pay a high price for insurance, when it would be more efficient for them to save money instead. One way to address this inconsistency would be to offer another option, in parallel with the insurance. Indeed, if insurance is the only option for households, the principal abuses her monopoly position. In our opinion, this result highlights the importance of regulating this type of market. If the state's interest is in combating fuel poverty, insurance seems to be a good option, but at the same time, alternative solutions must also be developed. This will therefore motivate the study of a model with a prepayment option in \Cref{sec:prepayment}.

\begin{figure}[H]
    \centering
    \includegraphics[scale=0.5]{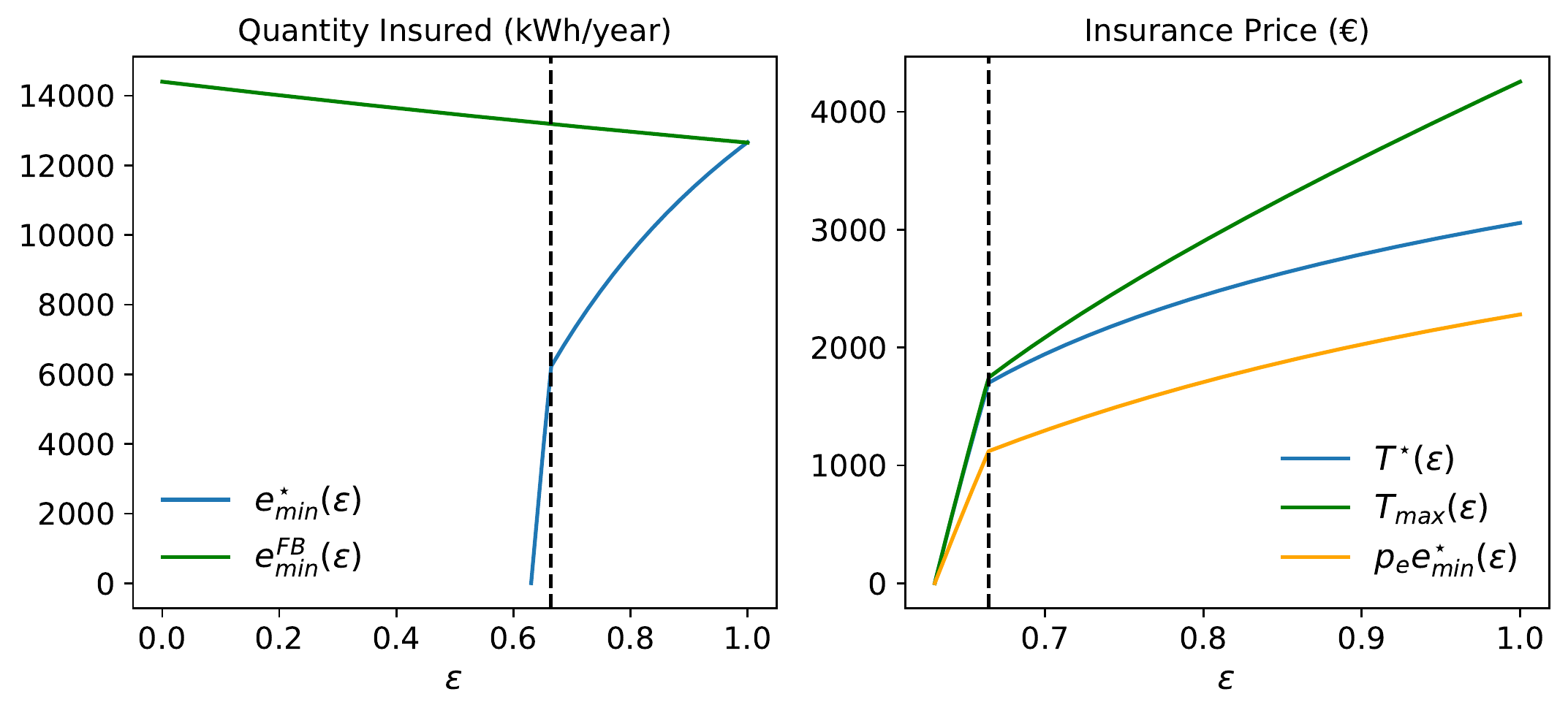}
\caption{Optimal insurance contract under adverse selection.}
{\footnotesize Left: the optimal insured quantity (blue) is compared to the quantity of the first-best case (green). Right: the optimal premium (blue) is compared to the maximum price (green) and to the future price (orange). Dotted black axes: $\varepsilon = \varepsilon_1^\star$.}
    \label{fig:quantity_SB}
\end{figure}

\subsection{Improvement of energy consumption}

Given the menu of contracts $(e_{\rm min}^\star, T^\star)$ associated with the optimal admissible revealing mechanism $(q^\star,t^\star_0)$, we can compute the optimal consumption of the agents who subscribe to an insurance contract (see \Cref{lem:no_insurance,lem:max_utility_t0,lem:max_utility_t1} for the formulas). In order to highlight the positive effects of quantity insurance, we compare with the quantities an agent will consume in the case of a financial insurance. In this case, the agent's problem is modified, in the sense that at time $t=1$, he no longer receives a quantity $e_{\rm min}$, but sees his income $\underline \omega w_0$ increases by a monetary amount $r$ (see \Cref{rk:financial_insurance}). This problem is highly similar to the one dealt with in this paper, which is why the theoretical results are omitted.

\medskip

The left graph of \Cref{fig:conso} represents the optimal consumption of energy at time $t=0$, $e_0^{\rm Q}$ with in-kind insurance (blue curve), $e_0^{\rm R}$ with classical revenue insurance (green curve) and $e_0^\varnothing$ without insurance (orange curve). Obviously, with insurance, an agent consumes less than he would without, since paying for insurance decreases the effective income he splits between the two goods.
Nevertheless, the in-kind insurance allows for a lower decrease in energy consumption compared to more traditional insurance.
Note that at time $t=1$, if the agent does not suffer a loss of income, the insurance is not activated, and he will thus consume the same quantity of energy in both cases, with and without insurance. 

\begin{figure}[h]
    \centering
    \includegraphics[scale=0.5]{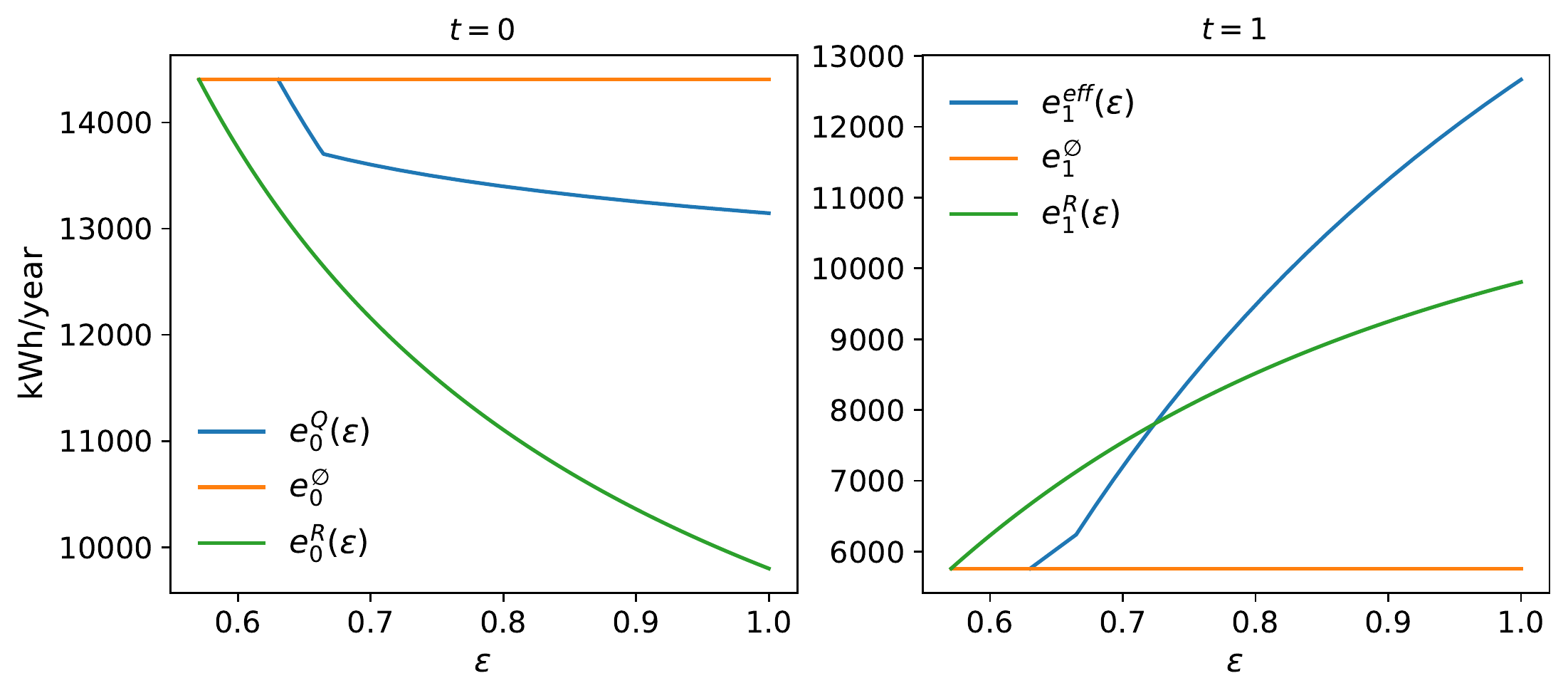}
\caption{Energy consumption in the third-best case.}
{\footnotesize Consumption with insurance (blue for in-kind, green for financial) and without (orange), at time $t=0$, on the left, and at time $t=1$ in the case of an income loss, on the right.}
    \label{fig:conso}
\end{figure}

In contrast, as we can see in the right graph, if the household suffers an income loss, he consumes more energy with insurance (blue and green curves) than without insurance (orange curve). More precisely, without insurance, his optimal consumption is $e_1^\varnothing \approx 5,800$ KWh, which is around the level of consumption of a household of four people without electric heating. Therefore, one can conclude that such a household would renounce home heating, for example.
Otherwise, if the agent of type $\varepsilon$ is insured, he receives either the quantity $e_{\rm min}^\star(\varepsilon)$ or a monetary compensation, depending on the insurance type.
In the case of the in-kind insurance, we can compute by \Cref{lem:max_utility_t1} his optimal consumption $e_1^{\rm Q} (\varepsilon)$ of energy. His effective consumption $e_1^{eff} (\varepsilon)$ (blue curve) is given by the sum of $e_1^{\rm Q}  (\varepsilon)$ and $e_{\rm min}^{\star}(\varepsilon)$. It is interesting to note that the higher the agent's risk type, the higher the insured quantity is, and it tends towards the quantity $e_0^\varnothing$ consumed with the initial income $w_0$. The results for a financial insurance are similar, but the increase in consumption for the riskiest households is smaller. 

\medskip

This latter fact can be explained as follows. By definition of $\varepsilon_1$, an agent of type $\varepsilon \in [\underline \varepsilon^{\star}, \varepsilon_1)$ receives a quantity $e_{\rm min}^\star (\varepsilon) < \widebar e_{\rm min}$. In this case, the insurance acts as an earmarked fund, or a liquid asset: the agent behaves exactly as if his income had been increased by $p_e e^\star_{\rm min}(\varepsilon)$. More precisely, a part $\alpha/(1+\alpha)$ of this supplementary income is dedicated to electricity consumption, and the other part, $1/(1+\alpha)$, is dedicated to the other good, in the same way that his income $w_t$ at time $t$ is distributed between the two goods.
In fact, since the insurance plan offers in-kind support, the agent has to decrease his consumption $e_1^{\rm Q}$ and increase his consumption $y_1^{\rm Q}$ to perfectly split this fictive supplementary income between the two goods.
In contrast, when the insured quantity is higher, precisely if $e^{\star}_{\rm min} (\varepsilon) > \widebar e_{\rm min}$, the agent's consumption $e_1^{\rm Q}$ is reduced to $0$, and he cannot decrease it anymore. In this case, his effective consumption is given by $e_{\rm min}^{\star}(\varepsilon)$, and the agent cannot properly split this fictive supplementary income between the two goods. For these types of agents, \textit{i.e.}, $\varepsilon \geq \varepsilon_1^\star$, the insurance no longer acts as a liquid asset, thus ensuring higher electricity consumption.

\medskip

These results on the energy consumption are very interesting in the situation we considered. Indeed, a household that falls into fuel poverty due to a loss of income tends to consume less electricity, which can lead to health problems and housing damage. This is exactly the kind of problem we are seeking to avoid by proposing an insurance. Nevertheless, traditional income insurance allows the agent to receive money in the event of loss of income. However, an agent in fuel poverty has other needs to satisfy that he considers more important. He would therefore use the insurance money largely for these expenses, ignoring the significance for his health of heating his home sufficiently, for example. The in-kind insurance we model prevents this bias; it manages to oblige the agent to consume enough electricity to live decently in his home. 

\subsection{Comparison with financial insurance in terms of utility}

Given the optimal contracts for both the quantity and financial insurances, we are able to compute numerically the expected utility of an household of any type $\varepsilon \in [0,1]$. Of course, less risky households have no interest in purchasing such policies, and therefore have a utility equal to their reservation utility. For riskier households, as we can see on the left graph of \Cref{fig:profit_SB}, both insurances allow an increase in utility, compared to the utility without insurance, but unfortunately, the utility obtained with quantity insurance is less than that obtained with income insurance. This is entirely due to the fact that, with in-kind insurance, the insured is forced to consume more energy than he would like for an equivalent income, and the benefit of this over-consumption is not taken into account in household utility. Indeed, the model assumes that households are not aware that under-consumption of energy can lead to long-term health problems, whose associated costs are paid by both the households and the society. 

\begin{figure}[H]
\centering
\includegraphics[scale=0.5]{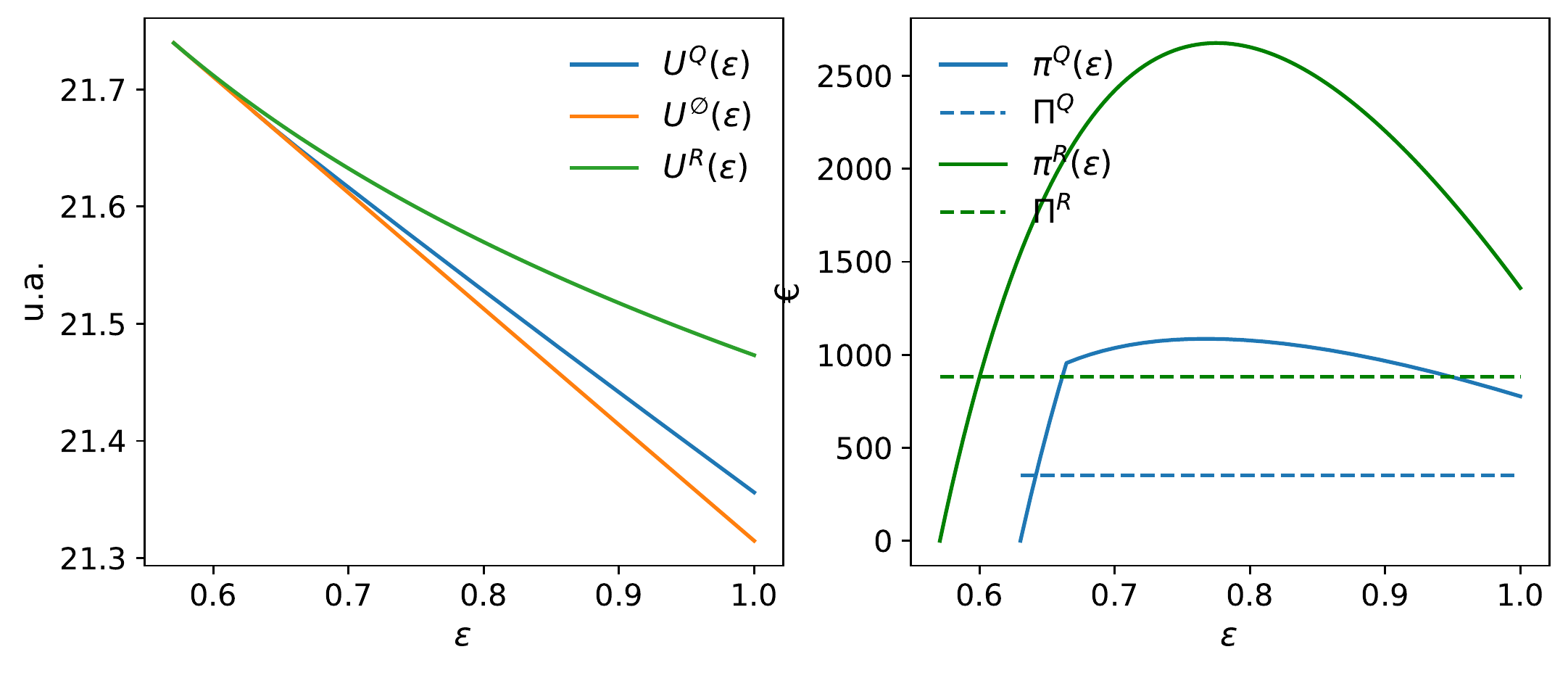}
\caption{Agent's utility and principal's profit under adverse selection.}
{\footnotesize Left: comparison of agent's utility with in-kind insurance (blue), financial insurance (green) and without insurance (orange). Right: principal's profit by type (plain curves) and expected profit (dotted lines) in both cases of an in-kind (blue) or financial (green) insurance.}
\label{fig:profit_SB}
\end{figure}

To account for the benefit of this increase in the energy consumption, one could take the point of view of a central planner, who is aware of the effects of fuel poverty. However, this would imply knowing how to measure the long-term costs of fuel poverty incurred by society as a whole, which is not a simple task. Indeed, first of all, it is difficult to directly link health problems to their intrinsic cause, and therefore to estimate the future cost of a precarious household.
Moreover, these costs are difficult to quantify for each type of agent, as they do not only depend on the types considered here but on many other parameters (for \textit{e.g.}, age, initial state of health, labour market status). Nevertheless, the direct and indirect costs incurred by fuel poverty are far from negligible since, as mentioned by \cite{chaton2020avoiding}, they represent more than one billion euros per year (see \cite{roys2010real} or \cite{ahrendt2016inadequate}).

\medskip

Finally, given the optimal admissible revealing mechanism $(q^\star,t^\star_0)$ for the quantity insurance, we can compute the profit $\pi^{\rm Q}$ of the principal for every type of agent, represented in the right graph of \Cref{fig:profit_SB} (blue curve). In particular, her expected profit $\Pi^{\rm Q}$ from the considered population, given by the integral of $\pi^{\rm Q}(\varepsilon)$ between $\underline \varepsilon^\star$ and $1$, is equal to $351$ \euro{}.
To show the benefit of considering a menu of contracts, we can compare our results to the profit induced by a unique contract. In this case, the optimal quantity and price can easily be computed: $e_{\rm min} \approx 10,538$ kWh/year and $T \approx 2,549$ \euro{}. This contract induces an average profit of $342$ \euro{} for the principal, and only agents of type $\varepsilon \geq 0.65$ subscribe to an insurance plan. Therefore, there is a positive gain for the principal in implementing a menu of contracts instead of a unique contract. More precisely, it represents an average gain of 2.7\% on each contract offered. Moreover, with a menu of contracts, more agents are insured, since $\underline \varepsilon^\star < 0.65$. 

\medskip

Similarly, in the case of a financial insurance, we can compute (numerically) the principal's profit $\pi^{\rm R}$ for every type of agent (see \Cref{fig:profit_SB}, green curve). In particular, her expected profit $\Pi^{\rm R}$ is equal to $883$ \euro{}. Therefore, proposing an in-kind insurance instead of a more classical revenue insurance implies a loss of utility from the point of view of both the household and the insurer. This result raises a new question: how could a central planner (\textit{e.g.} the state) encourage insurers to offer such in-kind insurance, which enables households to significantly improve their energy consumption? Indeed, a central planner, aware of the harmful effects of under-consumption of energy and wishing to reduce the risk of fuel poverty among the population, would be well advised to generalise the use of this type of insurance. In this case, we will be faced with a three-level principal-agent problem, where the state delegates to an insurance company the implementation of an in-kind insurance and remunerates it for this. Unfortunately, the problem would be more complicated and beyond the scope of this paper, especially since it would be necessary to estimate the long-term costs of fuel poverty, which, as mentioned above, are difficult to quantify. 
Therefore, in this paper, we focus only on the relationship between an insurance company and the households, and neglect the external aspects that lead the company to offer insurance with benefits in kind.
The aforementioned extension to take into account the role of a central planner, or a regulator, is set aside for the time being, awaiting future research.

\subsection{Adding the prepayment option}\label{sec:prepayment}

The results we developed throughout this paper highlight a certain form of irrationality of the agents, mainly due to their unwillingness to save money from one period to the next, and to the monopoly position of the insurer. Indeed, since the agent has only the choice between subscribing or not subscribing to an insurance plan and does not think about saving from one period to the next, he is willing to pay a very high price for the insurance. This section proposes a solution to address this problem of irrationality.

\medskip

We consider that a regulator offers (or forces the insurer to offer) another form of contract: the prepayment option. In this situation, the agent can $(i)$ subscribe at time $t=0$ to an insurance contract; $(ii)$ prepay a quantity $e^{\rm P}$, \textit{i.e.}, pay at time $t=0$ the price $p_e e^{\rm P}$ to receive $e^{\rm P}$ in $t=1$; or
$(iii)$ do nothing.
This prepayment option is a way to encourage a specific form of savings and should therefore limit the price of insurance. Although some theoretical results could be obtained, we choose in this section to only present numerical results. Indeed, the theoretical formulations are in the same spirit as those developed throughout the paper, so it seems more relevant and meaningful in our opinion to discuss only the results obtained numerically in within the framework of fuel poverty.

\subsubsection{A new reservation utility}\label{app:prepaiment}

If the agent decides to prepay at time $t=0$ a quantity $e^{\rm P}$, his utility function is defined by \eqref{eq:utility_t0}, where $T = p_e e^{\rm P}$ is the price of the chosen quantity. As previously, we assume without loss of generality that $e^{\rm P} = \alpha q \underline \omega w_0/p_e$ for some $ q \in \R_+$. Through easy optimisation techniques, we obtain the following result:
\begin{lemma}\label{lem:max_utility_PP_t0}
If the agent subscribes to the prepayment option for a quantity $q < 1/\alpha \underline \omega$, his maximum utility  at time $t=0$, given by $V_0^{\rm P} (q) = (1+\alpha) \ln (w_0) + (1+\alpha) \ln (1 - \alpha q \underline \omega) + C_{\alpha,p_e,p_y}$, can be achieved through the following consumption:
\begin{align*}
    y_0^{\rm P} := \dfrac{w_0}{1+\alpha} \dfrac{1 - \alpha q \underline \omega}{p_y} \; \text{ and } \; e_0^{\rm P} := \dfrac{\alpha w_0}{1+\alpha} \dfrac{1 - \alpha q \underline \omega}{p_e}.
\end{align*}
\end{lemma}

At time $t=1$, he receives the prepaid quantity $e^{\rm P}$ in any case, not only in the case of an income loss. His utility to maximise at time $t=1$ is thus naturally defined by:
\begin{align*}
    V^{\rm P}_1 (\omega, q) = \max_{(e_1, y_1) \in \R_+^2} \alpha \ln (e_1 + \alpha  q \underline \omega w_0 /p_e) + \ln (y_1), \; \text{ u.c. } \; e_1 p_e + y_1 p_y \leq \omega w_0.
\end{align*}
\begin{lemma}\label{lem:max_utility_PP_t1}
If the agent subscribes to the prepayment option for a quantity $q \in \R$, his maximum utility at time $t=1$, given by $V_1^{\rm P} (\omega, q) =
(1+\alpha) \ln (\omega w_0) + C_{\alpha,p_e,p_y} + \widebar U(q \underline \omega/ \omega)$, can be achieved through the following consumption:
\begin{align*} 
    e_1^{\rm P} := \dfrac{\alpha}{1+\alpha} \dfrac{ w_0}{p_e} \big( \omega  -  q \underline \omega \big)^+  \; \text{ and } \;
    y_1^{\rm P} := \dfrac{\omega w_0 - p_e e_1^{\rm P}}{p_y}.
\end{align*}
\end{lemma}

Therefore, by \Cref{lem:max_utility_PP_t0,lem:max_utility_PP_t1}, the expected utility of an agent of type $\varepsilon$ who chooses to prepay the quantity $e^{\rm P} = \alpha q \underline \omega w_0 / p_e$ for $q \in [0, 1/\alpha \underline \omega)$ is given by:
\begin{align*}
    \textnormal{EU}^{\rm P} (\varepsilon, q ) = &\ \textnormal{EU}^\varnothing (\varepsilon) + (1+\alpha) \ln \big(1 - \alpha q \underline \omega \big) + \beta \varepsilon \widebar U(q) + \beta (1-\varepsilon) \widebar U(q \underline \omega / \widebar \omega).
\end{align*}
The agent then chooses the optimal amount he wants to prepay by maximising his expected utility over an admissible $q$.
The easiest way to solve this optimisation problem is to perform a simple numerical optimisation to find the optimal quantity $q^{\rm P}(\varepsilon)$ that an agent of type $\varepsilon$ should prepay and his associated expected utility, denoted $\textnormal{EU}^{\rm P, \star} (\varepsilon)$, for every $\varepsilon \in [0,1]$. For the parameters previously defined in \Cref{ss:fuel_poverty_param}, the results are presented in \Cref{fig:prepayment}.

\begin{figure}[H]
\includegraphics[scale=0.5]{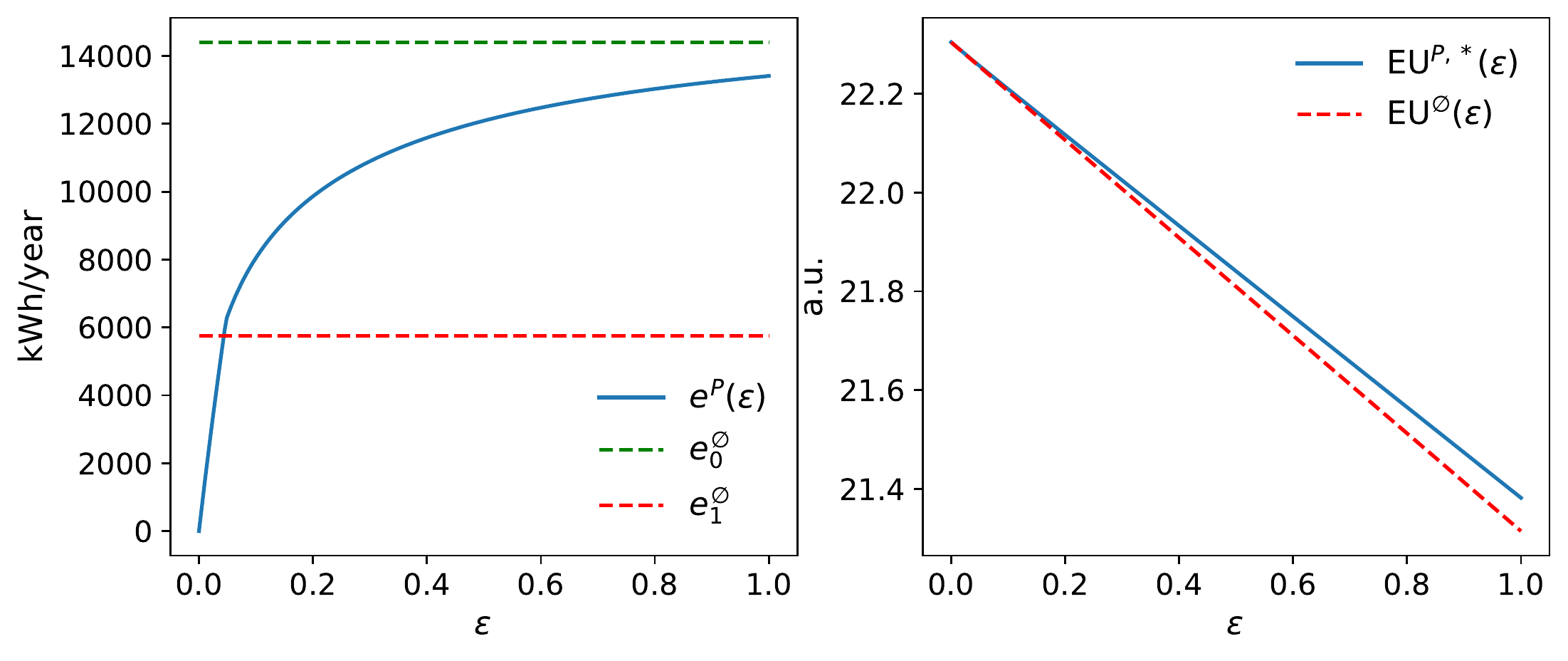}
\centering
\caption{Optimal quantity prepaid and associated expected utility.}
{\footnotesize Left: the optimal quantity prepaid (blue) is compared to the quantity consumed with the initial income (dotted green) and with an income loss (dotted red). Right: the expected utility with prepayment (blue) is compared to the previous reservation utility (dotted red).}
\label{fig:prepayment}
\end{figure}

Facing this new option, an agent of type $\varepsilon$ subscribes to an insurance contract $(e_{\rm min}, T)$, with $e_{\rm min} = \alpha q \underline \omega w_0/p_e$ and $T = w_0 t_0$, if and only if the two following conditions hold:
\begin{align}\label{eq:participation_prep}
    \textnormal{EU}^{\rm Q} (\varepsilon, q, t_0) \geq \textnormal{EU}^{\rm P, \star} (\varepsilon) \; \text{ and } \; 
    \textnormal{EU}^{\rm Q} (\varepsilon, q, t_0) \geq \textnormal{EU}^{\varnothing} (\varepsilon).
\end{align}
By definition of $\textnormal{EU}^{\rm P, \star} (\varepsilon)$, for every $\varepsilon \in [0,1]$, we have $\textnormal{EU}^{\rm P, \star} (\varepsilon) \geq \textnormal{EU}^{\rm P} (\varepsilon, 0) = \textnormal{EU}^{\varnothing} (\varepsilon)$. Therefore,
the second inequality in \eqref{eq:participation_prep} is implied by the first and is thus not necessary. In this framework, the reservation utility of an agent of type $\varepsilon$ is therefore defined by the utility he obtains thanks to the prepayment option. To simplify the notation, we denote by $\Delta \textnormal{EU}^{\rm P, \star}$ the difference between the expected utility with prepayment of an agent and his utility without:
\begin{align}\label{eq:exp_utility_prep_net}
    \Delta \textnormal{EU}^{\rm P, \star} (\varepsilon) := \textnormal{EU}^{\rm P, \star} (\varepsilon) - \textnormal{EU}^{\varnothing} (\varepsilon).
\end{align}

\subsubsection{Insights from the first-best case}\label{app:prepayment_FB}

The problem of the principal in this case is defined by \eqref{eq:def_principal_pb_FB}, under a new participation constraint of the agent, since his reservation utility is now given by $\textnormal{EU}^{\rm P, \star} (\varepsilon)$: an agent accepts the contract if it provides him with at least as much utility as the prepayment. Similar to the reasoning developed in \Cref{ss:principal_FB} for the initial problem, when the principal knows the type of the agent, she may charge him the highest price he is willing to pay for the insurance. In this case, the agents' informational rents are then reduced to zero for any type $\varepsilon \in [0,1]$.

\medskip

Using the notation \eqref{eq:exp_utility_prep_net}, the participation constraint of an agent of type $\varepsilon$ is equivalent to:
\begin{align*}
    t_0 \leq &\ t_{max}^{\rm P}(\varepsilon, q) := 
    1 - \exp \bigg( \dfrac{\Delta \textnormal{EU}^{\rm P, \star} (\varepsilon)}{1+\alpha}  \bigg) \times
    \left\{
    \arraycolsep=1.4pt\def\arraystretch{1.2}
        \begin{array}{ll}
        (1 + q \alpha)^{- \beta \varepsilon }
        & \mbox{ if } q < 1, \\
        q^{- \beta \varepsilon \frac{\alpha}{1+\alpha}} 
        ( 1+ \alpha)^{- \beta \varepsilon }
        & \mbox{ if } q \geq 1,
        \end{array}
    \right.
\end{align*}
However, since it is relatively complicated to explicitly obtain the expected utility with prepayment, we cannot give a more detailed formula for the maximum price that an agent of type $\varepsilon$ is willing to pay for insurance. Nevertheless, all the results can easily be computed numerically. 

\begin{figure}[H]
\includegraphics[scale=0.5]{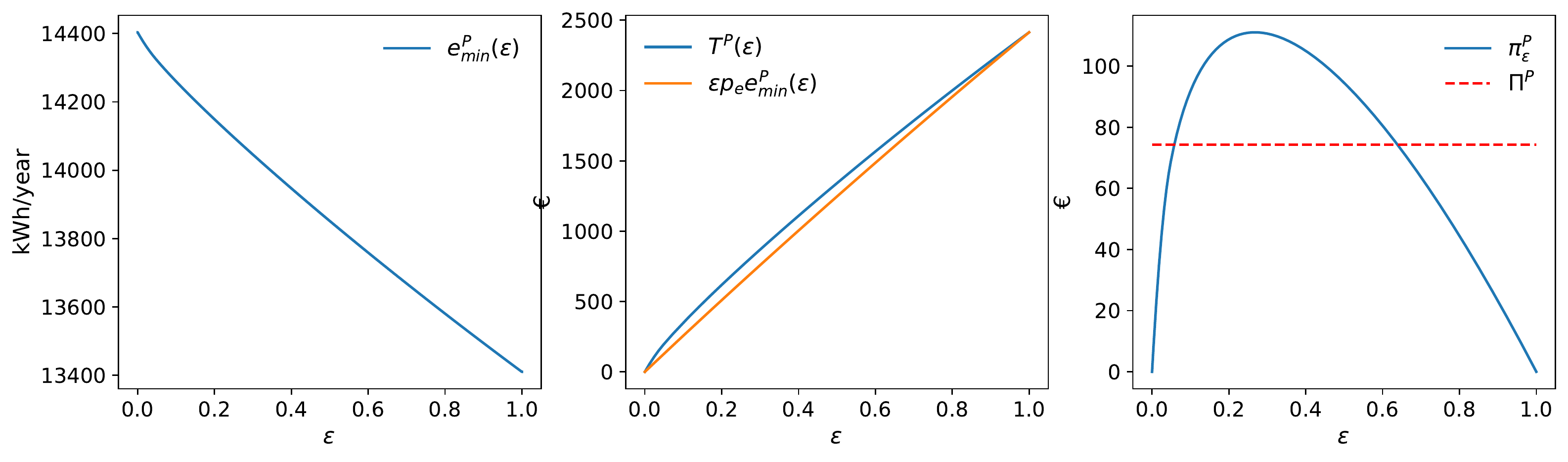}
\centering
\caption{Optimal insurance in the first-best case with prepayment.}
{\footnotesize Blue curves represent, from left to right: insured quantity, premium and principal's profit. Middle: premium is compared to actuarial price (orange curve), which also corresponds to the principal's cost. Right: red dotted line is her average profit.}
\label{fig:prepaymentFB}
\end{figure}

\Cref{fig:prepaymentFB} presents, from top to bottom, the insured quantity, the price paid by the agents, and the principal's profit, in the case of insurance against fuel poverty, \textit{i.e.}, with the parameters defined in \Cref{ss:fuel_poverty_param}. We can compare these graphs with those of \Cref{fig:FB_QI}. The most interesting point is that the price of insurance is significantly lower in this new situation. Indeed, the price paid by the agents is now barely higher than the actuarial price, whereas without prepayment, it is sometimes even higher than $p_e e_{\rm min}$, which is actually precisely the price of the prepayment.

\subsubsection{Third-best case}\label{app:prepayment_TB}

Since we are only changing the reservation utility, the results of \Cref{ss:RP_QI} remain true. In particular, the principal offers a menu of revealing contracts $(e_{\rm min}, T)$ defined by \Cref{cor:contract_emin}, such that an agent of type $\varepsilon$ chooses the quantity $e_{\rm min} = \alpha q(\varepsilon) \underline \omega w_0/p_e$ and pays the price $w_0 t_0(\varepsilon)$, given by \eqref{eq:TQ_emin_eps}, for a particular function $q$, which will be optimised by the principal.
The only thing that changes is the participation constraint, which is now given by
$\textnormal{EU}^{\rm Q} (\varepsilon, q, t_0) \geq \textnormal{EU}^{\rm P, \star} (\varepsilon)$.
Using the form of the price given by \eqref{eq:TQ_emin_eps}, this constraint is equivalent to
\begin{align*}
c_q  \geq 
\underline c^{\rm P}(\varepsilon) := \exp \bigg( \dfrac{\Delta \textnormal{EU}^{\rm P, \star} (\varepsilon)}{1+\alpha}  \bigg) \underline c (\varepsilon).
\end{align*}
The function $\underline c$, defined by \eqref{eq:underline_c}, is decreasing, implying in \Cref{ss:PC} that the participation constraint was satisfied for agents of type $\varepsilon$ above a specific level.
Unfortunately, in this case, it is not possible to precisely determine the variations of the function $\underline c^{\rm P}$, since we do not have an explicit form and, as we can see in \Cref{fig:prepayment}, the difference defined in \eqref{eq:exp_utility_prep_net} is increasing. 

\medskip

In the case with prepayment, it is therefore difficult to determine a monotonicity or even the variations of the information rent. With the help of the first-best case, we can still intuit that the information rent in the third-best case is increasing up to a specific $\varepsilon \in [0,1]$ and then decreasing. More precisely, as seen on the graph with the principal's profit, she earns money on the agents of medium risk type, since agents of type $\varepsilon = 0$ and $\varepsilon = 1$ have no interest in subscribing to insurance plans. Indeed, on the one hand, the problem of agents of the lowest risk type 
remains the same as in the case without prepayment: the optimal quantity they would like to prepay is zero, which implies that their utility with prepayment is equal to their utility without insurance. Their reservation utility is therefore unchanged from the case studied throughout this paper, and these agents are not of interest from the insurer's point of view. On the other hand, the agents of type $\varepsilon =1$ are now indifferent between prepayment and insurance, the two options providing them with a sufficient quantity of energy for the future. The agents of the highest risk type, who were highly courted by the insurer, are now hard to satisfy and instead turn to prepayment. Hence, the most `interesting' agents in this case for the Principal seem to be the intermediate ones.

\medskip

Unfortunately, to address the third-best case in a rigorous manner, further study is required. The reasoning would be similar to the one developed in \Cref{sec:third_best}, but the non-monotonicity of the information rent makes the problem more difficult. Nevertheless, following the intuition previously developed, we can make the reasonable assumption that the optimal contract should only select types $\varepsilon$ in an interval $[\underline \varepsilon, \widebar \varepsilon] \subset [0,1]$. 
We can then apply the reasoning developed in \Cref{sec:third_best} to find the ODE satisfied by the optimal contract under this assumption, which is highly similar to \eqref{eq:ODE}, and thus not detailed here. We only present the numerical results, obtained with the usual parameters. Especially, \Cref{fig:prepaymentSB} presents the optimal quantity insured (left graph) as well as the optimal premium (right). 

\begin{figure}[H]
    \centering
    \includegraphics[scale=0.5]{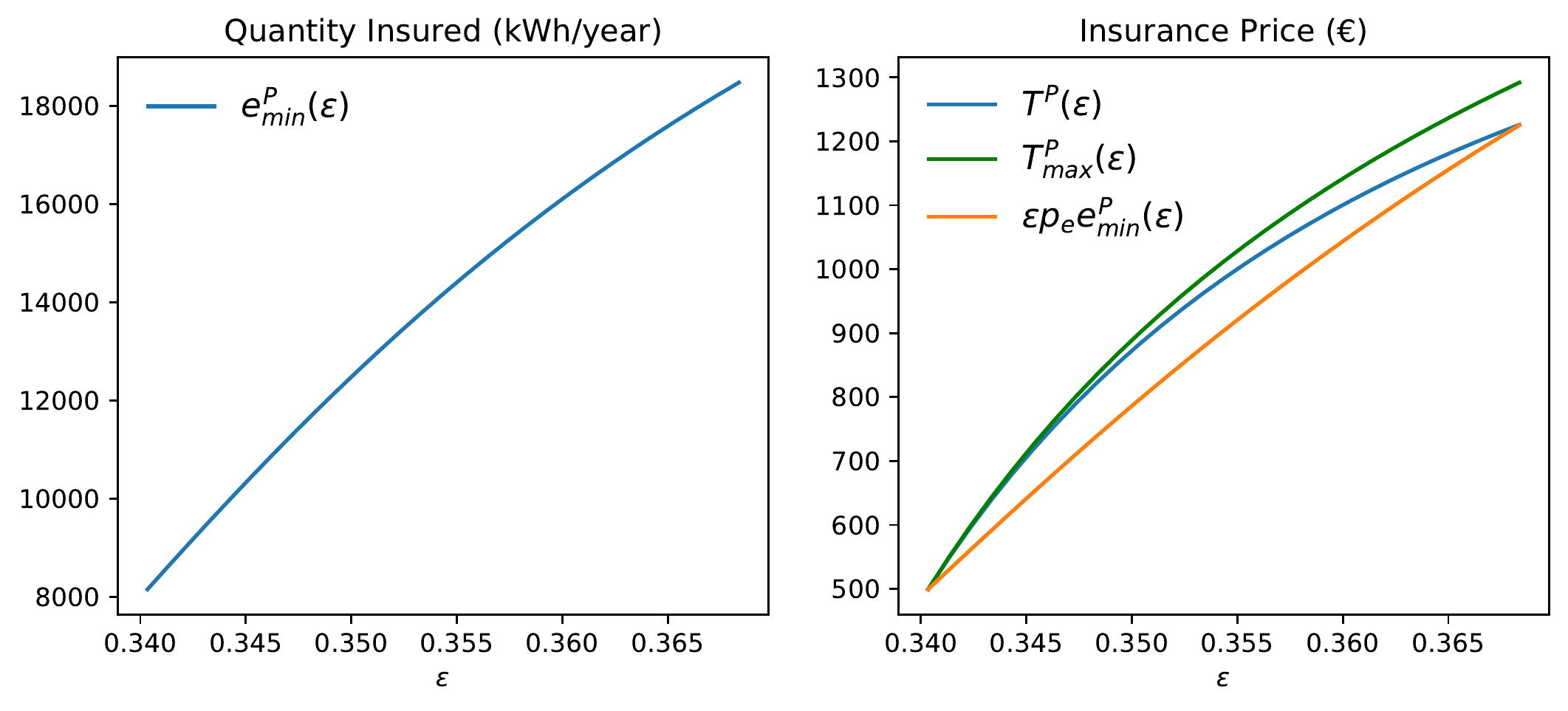}
\caption{Insurance contract in the third-best case, when the prepayment option is available.}
{\footnotesize Left: optimal insured quantity. Right: optimal premium (blue), compared to maximum price (green) and to actuarial price (orange).}
\label{fig:prepaymentSB}
\end{figure}

Comparing these results with those presented in \Cref{fig:quantity_SB}, it is clear that, when the prepayment option is available:
\begin{itemize}
    \item the insured quantity is higher, thus clearly ensuring a sufficient level of energy consumption at time $t=1$, which should contribute in reducing the risk of fuel poverty among the households;
    \item the premium of the insurance is just slightly above the actuarial price, given by $\varepsilon p_e e_{\rm min}^{\rm P}$ for an agent of type $\varepsilon$, and therefore way below the future price of the quantity.
\end{itemize}
Therefore, the addition of the prepayment option seems to $(i)$ allow agents to consume sufficient energy in case of loss of income, either by insurance or prepayment; $(ii)$ decrease the insurance premium compared to the case without prepayment; and $(iii)$ allow agents of a medium risk type to be insured.

\medskip

However, it can be seen that the contract we obtain is designed only for a small proportion of types, approximately $2\%$. From this point of view, the value of implementing such insurance is largely debatable, and the prepayment option seems to be a more relevant and sufficient tool to limit fuel poverty. Nevertheless, it should be noticed that, due to the assumption we have made, a contract with more than one interval could be more optimal than this one, and would thus select more types of agents. Secondly, it is worth noticing that the interval depends on the parameters considered. For example, if we focus on slightly wealthier households ($w_0 = 50,000$ instead of $35,000$, and $\alpha = 0.06$) with better expected income at time $t=1$ ($\overline \omega = 1.5$ instead of $1$), the proportion of selected types increases up to almost $10\%$. Finally, we only considered the case of a uniform distribution of types in the population. It is likely that, if the distribution is different, particularly if the average risk agents represent a larger part of the population, then the selected types would represent a more significant proportion.

\medskip

Nevertheless, these results still underline the need for regulation of this market. More precisely, in this case, a regulator could force or try to incentivise the insurance company to offer more different contracts, in order to cover a larger proportion of households. In this perspective, our work could help quantify the loss of income for the insurer if she has to offer contracts for a wider range of risk types.

\section{Conclusion}\label{sec:conclusion}

We develop a two-period principal-agent problem with adverse selection and endogenous reservation utility to model insurance with benefits in kind. This model allows us to obtain semi-explicit solutions. Applied to the energy sector, this in-kind support helps to prevent fuel poverty among households. 
Indeed, when a household suffers from a loss of income, if it has subscribed to an insurance contract of the sort we propose, it consumes more energy than it would without insurance. In this application, providing support in kind therefore forces the household to consume more energy and thus prevents risky behaviour that can lead to serious health problems. Therefore, even in the least efficient scenario from the households point of view, \textit{i.e.}, when the insurance is provided by a monopoly, this in-kind mechanism improves significantly the living conditions of the riskiest households by ensuring them a sufficient consumption of energy. An insurance mechanism also makes it possible to cover agents' fuel poverty risks and to pool costs between agents (of a high risk type). 
Through a slight adaptation of our model, we show that the previous conclusions does not hold in the case of a classic income insurance: the household would not increase its energy consumption sufficiently. This result is inline with those of \citet*{slesnick1996consumption}, who underlines that
\blockquote[p. 1544][]{[w]ith the exception of energy and consumer goods, in-kind transfers have roughly the same impact on the poverty rate as cash}.
With these results in mind, and in our opinion, developing such insurance could be a short-term solution to prevent households from experiencing fuel poverty, and thus avoid significant societal costs.

\medskip

Nevertheless, our study raises various problems and questions, notably on the regulation of monopoly, which could constitute relevant extensions of our model. 
One of the main issues is that, when computing the utility of the household and the profit of the principal, it appears that a financial insurance would be more beneficial to both. This problem stems from the fact that, in our model, both the agent and the principal neglect the costs associated with fuel poverty.
However, while the effects of fuel poverty on the physical and mental health of individuals are not questionable (see \cite{lacroix2015fuel}), the societal cost of fuel poverty is difficult to quantify, which is why it is not taken into account in our model. A possible extension could, therefore, be to consider the problem of a central planner (\textit{e.g.} the State), who is aware of the harmful effects and societal costs of under-consumption of energy, and wish to reduce the risk of fuel poverty among the population. 
This central planner could therefore encourage insurers to offer such in-kind insurance, which enables households to significantly improve their energy consumption. 
In this perspective, consideration could also be given to making this type of insurance mandatory, either by law or through a contract between a landlord and a tenant, for example. Indeed, on the one hand, to reduce health expenses, the State could have an interest in ensuring that agents do not experience fuel poverty. On the other hand, a landlord, or even a social landlord, could have the same interest in order to prevent the tenant from causing deterioration of the housing.

\medskip

The second issue lies on the fact, under a monopolistic framework and without alternative option, the insurance premium is high. Indeed, since the households considered do not think about saving to ensure a sufficient quantity of essential goods regardless of their future income, their choice is limited to subscribing to the insurance plan or doing nothing. Due to the shape of the chosen utility function, in particular its concavity, agents who anticipate a loss of income with a high probability are willing to pay a very high price for the insurance we offer,
sometimes even higher than the future price of the insured quantity. Therefore, insurance can be a tool to protect at-risk households, but it cannot be set up alone if securing the welfare of the agents is the desired goal. 
Some suggestions could be considered to make the insurance premium more realistic. First, if a regulator required insurance to be offered to all types of agents, \textit{i.e.}, $\underline \varepsilon=0$, the insurance premium would be lower.
In addition, the regulator could also better control the monopoly position of electricity suppliers in the case of energy insurance. By introducing competition in the market for this type of insurance, the price should fall towards the marginal cost of insurance.

\medskip

The suggestion we propose in this paper to address the previous issue is to offer to the household the opportunity to engage in prepayment. We show that the implementation of this alternative will allow a significant decrease of the insurance premium.
However, the prepayment option highlights a third problem: in this case, the insurance company drastically limits the proportion of the population for which insurance is attractive.
This result still underline the need for regulation of this market. More precisely, in this case, a regulator could, for example, incentivise the insurance company to offer more different contracts, in order to cover a larger proportion of households. In this perspective, our work could help quantify the loss of income for the insurer if she has to offer contracts for a wider range of risk types.

\medskip

In addition, one could envisage more straightforward extensions. For example, the simplicity of our model makes it easy to extend into a multi-period model, keeping in mind that only agents who are not in precarious situations are entitled to subscribe to insurance plans. Thus, an agent who has not suffered a loss of income can again pay the premium to re-insure himself for the next period. A simple repetition of the model is sufficient to address this case. Moreover, our model could be extended to random prices, which would allow the household to have a guaranteed quantity of electricity even in case of a price increase.
Finally, while we have chosen to apply it for insurance against fuel poverty, this type of model could be used for other staples with consumption that is affected by a loss of income. More generally, this type of insurance could also be purchased by manufacturing firms to ensure that they have sufficient input in the event of a temporary downturn in revenues.

{\small \bibliography{bibliographyFuelPoverty.bib}}

\newpage

\begin{appendices}

\renewcommand{\appendixpagename}{Appendix: for online publication}

\appendixpage

\crefalias{section}{appendix}
\crefalias{subsection}{appendix}

\section{Complements on agent's problem}\label{app:agent_pb}

\subsection{Optimal consumption with insurance}\label{app:with_insurance}

For better readability, this appendix regroups the results obtained on optimal consumption in each good and at each period of time with insurance. These results are not at the heart of our study, but they represent necessary steps to establish \Cref{prop:exp_utility} that define the expected utility of an agent with insurance. Moreover, the results on optimal consumption in each case are used in numerical simulation to compare the effect of insurance on consumption.

\begin{lemma}[With insurance, at time 0]\label{lem:max_utility_t0}
Given an insurance premium $T$, the agent can achieve his maximum utility $V_0 (T) = (1+\alpha) \ln (w_0 - T) + C_{\alpha,p_e,p_y}$ at time $t=0$ through the following consumption in each good:
\begin{align*}
    y_0^{\rm Q} := \dfrac{1}{1+\alpha} \dfrac{w_0 - T}{p_y} \; \text{ and } \; e_0^{\rm Q} := \dfrac{\alpha}{1+\alpha} \dfrac{w_0 - T}{p_e}.
\end{align*}
\end{lemma}

Without loss of generality, we can assume that $T$ is of the form $T := t_0 w_0$ for $ t_0 \in [0,1)$.
Therefore, by slightly abusing the notations as in \Cref{lem:no_insurance}, we denote by $V_0 (t_0)$ the maximum utility the agent can achieve by optimally choosing his consumption, which can be written as
\begin{align}\label{eq:max_utility_t0}
    V_0 (t_0) &= (1+\alpha) \ln (w_0) + (1+\alpha) \ln \big(1 - t_0 \big) + C_{\alpha,p_e,p_y}.
\end{align}
The agent thus pays a price $w_0 t_0$ at time $t=0$, depending on the amount of staple good $e_{\rm min}$ that he wants to receive at time $t=1$.
By setting the quantity $\widebar e_{\rm min} := \alpha \underline\omega w_0/p_e$, we obtain the following result on optimal consumption at time $t=1$.
\begin{lemma}[With insurance, at time 1]\label{lem:max_utility_t1}
Given an insurance contract $(e_{\rm min}, T)$, the optimal agent's consumptions of each good at time $t=1$ are given by
\begin{align*} 
    e_1^{\rm Q} := \bigg( \dfrac{\alpha}{1+\alpha} \dfrac{\omega w_0}{p_e} - \dfrac{1}{1+\alpha} e_{\rm min} \mathds{1}_{\omega = \underline\omega} \bigg)^+  \; \text{ and } \;
    y_1^{\rm Q} := \dfrac{\omega w_0 - p_e e_1^{\rm Q}}{p_y},
\end{align*}
where $x^+ := \max \{x, 0\}$ for all $x \in \R$ and provide the following maximum utility to the agent:
\begin{align*}
    V_1 (\omega w_0, e_{\rm min}) &= 
        \left\{
        \begin{array}{ll}
        (1+\alpha) \ln (\omega w_0 + p_e e_{\rm min} \mathds{1}_{\omega = \underline\omega}) + C_{\alpha,p_e,p_y} & \mbox{ if } e_{\rm min} \mathds{1}_{\omega = \underline \omega} < \widebar e_{\rm min}, \\
        \ln (\omega w_0) + \alpha \ln (e_{\rm min}) - \ln (p_y) & \mbox{ if } e_{\rm min} \mathds{1}_{\omega = \underline\omega} \geq \widebar e_{\rm min}.
        \end{array}
    \right.
\end{align*}
\end{lemma}

The case separation in the previous proposition is needed to ensure that the consumption $e^{\rm Q}_1$ at time $t=1$ is non-negative. Indeed, the consumer should not be allowed to sell back the staple good.
In the first case, \textit{i.e.}, when $e_{\rm min} \mathds{1}_{\omega = \underline\omega} < \widebar e_{\rm min}$, the agent's utility at time $t=1$ depends only on his effective income $\omega w_0 + p_e e_{\rm min} \mathds{1}_{\omega = \underline\omega}$. Assuming that the choice of $e_{\rm min}$ is restricted to the interval $[0, \widebar e_{\rm min}]$ is equivalent to assuming that the quantity offered is smaller than the optimal quantity consumed in the event of an income loss. Therefore, in this case, the insurance acts as an earmarked fund or a liquid asset: the agent would have spent at least the quantity $e_{\rm min}$, so he reacts as if his income were increased by this value. In contrast, in the second case, the agent consumes only the amount $e_{\rm min}$ of the staple good: the optimal $e_1^{\rm Q}$ is equal to zero. In this case, the insurance is not interpreted as a liquid asset, and the agent consumes the full quantity $e_{\rm min}$ offered to him, leading to a utility of $ \alpha \ln (e_{\rm min})$, and does not consider it an increase in income. He then spends all his income $\underline\omega w_0$ on the other good.

\medskip

To simplify the notations, we can assume without loss of generality that $e_{\rm min}$ is of the form $e_{\rm min} = q \alpha \underline\omega w_0 /p_e$ for some $q \in \R_+$. The maximum utility obtained by the agent at time $t=1$ can thus be written as a function of $\omega$ and $q$ as follows:
\begin{align}\label{eq:max_utility_t1}
    V_1 (\omega, q) &=
    (1+\alpha) \ln (\omega w_0) + \widebar U(q) \mathds{1}_{\omega = \underline\omega} + C_{\alpha,p_e,p_y},
\end{align}
where $\widebar U$ is defined by \eqref{eq:overline_U}.
Combining \eqref{eq:max_utility_t0} and \eqref{eq:max_utility_t1}, we can explicitly compute the expected utility of an agent subscribing to an insurance contract, which allows us to state \Cref{prop:exp_utility}. Then, comparing the utility with and without insurance, we can determine when an agent of type $\varepsilon$ subscribes to the insurance plan (see \Cref{prop:PC}). 

\subsection{On the participation constraint}\label{app:tech_proof_CS}

The agent thus subscribes to the insurance plan as soon as the premium is below a specific level, given by \eqref{eq:max_price_QI}. Given the form of the maximum price, it can already be noted that some agents show a certain form of irrationality due to their unwillingness to save money from one period to the next.

\begin{remark}[Maximum price without uncertainty]\label{rk:without_uncertainty_QI}
One may note that, in our framework, the maximal price the consumer is willing to pay in the case without uncertainty is not equal to the actuarial price, $p_e e_{\rm min}$. Indeed, assuming that $q < 1$ and setting $\beta = 1$ for simplicity, we obtain $w_0 t_{max} (1, q) > p_e e_{\rm min}$ as soon as $w_0 > \underline\omega w_0 + p_e e_{\rm min}$. Therefore, if the income of the consumer at time $t=0$ is larger than the effective money he has at time $t=1$ with the insurance, he is willing to pay a certain amount of money at time $t=0$ to obtain less at time $t=1$. Conversely, if his income $w_0$ is lower than the money he will receive at time $t=1$, he is not willing to pay at time $t=0$ the real price of the energy he will get at time $t=1$. 
This result is slightly counterintuitive but is totally explained by the choice of concave utilities in a two-period model and the absence of saving.
This problem does not occur in single-period models.
However, a one-period model would not allow us to model a household willing to insure against a possible loss of future income. One solution could be to offer the agent the opportunity to have savings, but this is not consistent with the type of household being considered. Therefore, an alternative approach to address this issue is initiated in \textnormal{\Cref{sec:prepayment}}.
\end{remark}

\begin{proof}[\Cref{cor:comp_stat_tmax}]
To establish the comparative statics detailed in the corollary, it suffices to compute the relevant derivatives of $T_{max}$, given by \eqref{eq:max_price_QI_emin}.

\medskip

In the one hand, when $e_{\rm min} < \alpha \underline\omega w_0/p_e$, we have:
\begin{align*}
    \dfrac{\partial_{e_{\rm min}} T_{max} (\varepsilon, e_{\rm min})}{ w_0 - T_{max} (\varepsilon, e_{\rm min})} &=  
    \dfrac{\beta \varepsilon p_e}{\underline\omega w_0 + p_e e_{\rm min}}, \quad
    \dfrac{\partial_{p_e} T_{max} (\varepsilon, e_{\rm min})}{ w_0 - T_{max} (\varepsilon, e_{\rm min})}
     =
    \dfrac{\beta \varepsilon e_{\rm min} }{\underline\omega w_0 + p_e e_{\rm min} } \big( w_0 - T_{max} (\varepsilon, e_{\rm min}) \big), \\
    \dfrac{\partial_{\varepsilon} T_{max} (\varepsilon, e_{\rm min})}{ w_0 - T_{max} (\varepsilon, e_{\rm min})}
     &= 
    \beta \ln \bigg(1 + \dfrac{p_e e_{\rm min}}{\underline\omega w_0} \bigg), \quad
    \dfrac{\partial_{\beta} T_{max} (\varepsilon, e_{\rm min})}{ w_0 - T_{max} (\varepsilon, e_{\rm min})}
    = 
    \varepsilon \ln \bigg(1 + \dfrac{p_e e_{\rm min}}{\underline\omega w_0} \bigg), \\
    \text{and } \quad
    \dfrac{\partial_{\underline \omega} T_{max} (\varepsilon, e_{\rm min})}{ w_0 - T_{max} (\varepsilon, e_{\rm min})}
     &= 
    \dfrac{- \beta \varepsilon p_e e_{\rm min}}{\underline\omega \big( \underline\omega w_0 + p_e e_{\rm min} \big)}.
\end{align*}
It is easy to see that these derivatives are non-negative, except the last one which is non-positive, since $T_{max}$ is bounded by $w_0$, and all other parameters are non-negative. Furthermore, we have
\begin{align*}
    \partial_{w_0} T_{max} (\varepsilon, e_{\rm min}) &= 1 
    - \exp \bigg( - \beta \varepsilon \ln \bigg(1 + \dfrac{p_e e_{\rm min}}{\underline\omega w_0} \bigg) \bigg)
    \bigg( 1 + \dfrac{\beta \varepsilon  p_e e_{\rm min} }{ \underline\omega w_0 + p_e e_{\rm min} } \bigg).
\end{align*}
Further computations are needed to ensure that this derivative is non-negative. Let us introduce the following function, for all $x \in [0,1)$:
\begin{align*}
    f(x) &= 1 - \erm^{- \beta \varepsilon \ln (1 + x)}
    \bigg( 1 + \dfrac{\beta \varepsilon  x }{1 + x} \bigg).
\end{align*}
Noticing that the following equality holds:
\begin{align*}
    \partial_{w_0} T_{max} (\varepsilon, e_{\rm min}) = f \bigg( \frac{p_e e_{\rm min}}{\underline\omega w_0} \bigg),
\end{align*}
it is only necessary to prove that $f$ is non-negative for all $x \in [0,1)$ in order to obtain the desired result, namely that $T_{max}$ is increasing with respect to $w_0$. First, remark that $f(0) = 0$ and that the derivative of $f$ satisfies:
\begin{align*}
    f'(x) &=
    \dfrac{\beta \varepsilon x (1+ \beta \varepsilon)}{(1+x)^2}
    \erm^{- \beta \varepsilon \ln (1 + x)} \geq 0 \; \text{ for all } \; x \in [0,1).
\end{align*}
Therefore, the function $f$ starts at $0$ for $x = 0$ and is then increasing for $x \in [0,1)$, which implies that $f$ is non-negative on $[0,1)$. We thus conclude that $T_{max}$ is increasing with respect to $w_0$.

\medskip

In a similar way, we can show that, when $e_{\rm min} \geq \alpha \underline\omega w_0/p_{\rm e}$, the derivatives have the same sign as their counterparts on the previous interval, confirming that $T_{\rm max}$ is increasing with respect to $w_0$, $e_{\rm min}$, $p_{\rm e}$, $\varepsilon$ and $\beta$, but decreasing with respect to $\underline \omega$. 
\end{proof}

\section{Benchmark case: the first-best problem}\label{sec:first_best}

For the sake of completeness, we solve in this section the problem in the first-best case, \textit{i.e.}, without adverse selection. In particular, since the principal knows the type of the agent, she can offer him a specific contract, with which his participation constraint is binding. In other words, she may charge the insurance at the highest price the agent is willing to pay, to the point that he is, in fact, indifferent between subscribing or not subscribing to the contract. As already mentioned, this section gives some insights on what should happen under adverse selection. 

\subsection{Solving the principal's problem}\label{ss:principal_FB}

As detailed in \Cref{ss:principal}, the principal's problem in the first-best case is defined by \eqref{eq:def_principal_pb_FB}, under the participation constraint of the agent. Thanks to the reasoning developed in the previous subsection and denoting by $\Xi_{\varepsilon} := \{ (q,t_0) \in \R_+^2, \text{ s.t. } t_0 \leq t_{max} (\varepsilon, q) \}$, her problem is equivalent to:
\begin{align}\label{eq:principal_pb_FB_simple}
    \pi_{\varepsilon} := w_0 \; \sup_{(q, t_0) \in \Xi_{\varepsilon}} \big( t_0 - \varepsilon \alpha q \underline\omega \big).
\end{align}

\begin{proposition}\label{prop:pb_principal_FB}
If $\beta > \underline\omega$, the optimal contract $(e_{\rm min}, T)$ for an Agent of type $\varepsilon \in [0,1]$ is given by $e_{\rm min}^{FB} := \alpha q_\varepsilon \underline\omega w_0/p_e$ and $T^{FB} := w_0 t_{max} (\varepsilon, q_\varepsilon)$ where
\begin{align}\label{eq:q_optimal_FB}
q_\varepsilon := 
\left\{
\arraycolsep=1.4pt\def\arraystretch{1.8}
    \begin{array}{ll}
    \Big( \beta (1+\alpha)^{-\beta \varepsilon-1} /\underline\omega \Big)^{(1+\alpha)/(1+\alpha+\alpha \beta \varepsilon)} & \mbox{ if } \varepsilon \leq \varepsilon_1^{FB}, \\
    \dfrac{1}{\alpha} \Big( \big( \beta/\underline\omega \big)^{1/(1+\beta \varepsilon)} - 1 \Big) & \mbox{ if } \varepsilon > \varepsilon_1^{FB},
    \end{array}
\right.
\end{align}
for $\varepsilon_1^{FB} := \dfrac{1}{\beta} \bigg( \dfrac{\ln (\beta) - \ln (\underline\omega)}{\ln (1+ \alpha)} - 1 \bigg)$.
\end{proposition}

The profit of the principal in the first-best case, induced by the optimal contract detailed in the previous proposition, is given by the following result.
\begin{corollary}\label{cor:pb_principal_FB}
Let us assume $\beta > \underline \omega$. The optimal contract $(e_{\rm min}^{FB}, T^{FB})$ for an agent of type $\varepsilon \in [0,1]$ induced the following profit for the principal:
\begin{align*}
\pi_\varepsilon^{FB} := 
\left\{
\arraycolsep=1.4pt\def\arraystretch{1.6}
    \begin{array}{ll}
    w_0 \Big( 1 - \big( \beta (1+\alpha) / \underline\omega \big)^{- \beta \varepsilon/(1+\alpha + \alpha \beta \varepsilon)}  - \varepsilon \underline\omega \big( \beta (1+\alpha)^{-\beta \varepsilon} / \underline\omega \big)^{(1+\alpha)/(1+\alpha+\alpha \beta \varepsilon)} \Big) 
    & \mbox{ if } \varepsilon \leq \varepsilon_1^{FB}, \\
    w_0 \Big( 1 + \varepsilon \underline \omega - \big(\beta / \underline\omega \big)^{-\beta \varepsilon /(1+\beta \varepsilon)} - \varepsilon \underline\omega \big(\beta /\underline\omega \big)^{1/(1+\beta \varepsilon)} \Big) 
    & \mbox{ if } \varepsilon > \varepsilon_1^{FB}.
    \end{array}
\right.
\end{align*}
\end{corollary}

The proof of the previous corollary results from the proof of the associated proposition, detailed below.
\begin{proof}[\Cref{prop:pb_principal_FB}]
We fix the agent's type $\varepsilon \in [0,1]$. Since the profit of the principal is increasing in $T$, she has an incentive to set the price of the insurance equal to the maximum price the agent is willing to pay, \textit{i.e.}, $t_0 = t_{max} (\varepsilon, q)$, for $q \in \R_+$. %
The participation constraint of the agent is thus binding, and the maximisation problem of the principal \eqref{eq:principal_pb_FB_simple} becomes:
\begin{align*}
    \pi_{\varepsilon} = w_0 \max \bigg\{ 
    \sup_{q \in [0,1)} \Big\{ 1 - (1 + \alpha q)^{- \beta \varepsilon} - \varepsilon \alpha q \underline\omega \Big\}, \;
    \sup_{q \geq 1} \Big\{ 1 - q^{- \beta \varepsilon \frac{\alpha}{1+\alpha}} (1+ \alpha)^{- \beta \varepsilon} - \varepsilon \alpha q \underline\omega \Big\} \bigg\}.
\end{align*}
Computing the first- and second-order conditions (FOC and SOC) for each supremum, and since $\beta > \underline \omega$, we obtain that the two suprema are attained for $q_1$ and $q_2$, respectively, where:
\begin{align*}
    q_1 = \min \bigg\{ \dfrac{1}{\alpha} \Big( \big( \beta/\underline\omega \big)^{1/(1+\beta \varepsilon)} - 1 \Big), 1 \bigg\}, \; \text{ and } \; 
    q_2 = \max \bigg\{
    \Big( \beta (1+\alpha)^{-\beta \varepsilon} /\underline\omega \Big)^{(1+\alpha)/(1+\alpha+\alpha \beta \varepsilon)}, 1 \bigg\}.
\end{align*}
If $\varepsilon > \varepsilon_1^{FB}$, then $q_1 < 1$ and $q_2 = 1$ and, conversely, if $\varepsilon < \varepsilon_1^{FB}$, then $q_1 = 1$ and $q_2 > 1$. Since the two suprema have the same value for $q=1$, we conclude that $q_\varepsilon$ defined by \eqref{eq:q_optimal_FB} is optimal.
\end{proof}

\begin{remark}\label{rk:beta_omegal}
We assume in \textnormal{\Cref{prop:pb_principal_FB}} and \textnormal{\Cref{cor:pb_principal_FB}} that $\beta > \underline\omega$ because it is the most interesting case. Otherwise, we would have $\varepsilon_1^{FB} < 0$, and the maximum would be reached for $q_1$ defined in the previous proof. However, in this particular case, $q_1$ is negative for all $\varepsilon \in [0,1]$. Therefore, the optimal $q_\varepsilon$ is zero in this case for all $\varepsilon \in [0,1]$. This means that the principal has no interest in offering insurance. Indeed, when $\beta$ is too small, the agent has very little concern for his future, so he is not willing to pay for insurance to protect himself.
\end{remark}

\subsection{Some comparative statics}\label{ss:comparative_FB}

We first study the optimal insured quantity as well as the optimal insurance premium in the first-best case, using Equation \eqref{eq:q_optimal_FB}.
\begin{corollary}
\begin{enumerate}[label=$(\roman*)$]
\item The quantity $e_{\rm min}^{FB}$ is increasing with respect to $w_0$ and decreasing with respect to $p_e$ and $\varepsilon$. 
\item The price $T^{FB}$ is increasing with respect to $w_0$, $\varepsilon$ and $\beta$, but decreasing with respect to $\underline \omega$. Moreover, $T^{FB}$ is independent of $p_e$.
\item When $\varepsilon > \varepsilon_1^{FB}$, both $e_{\rm min}^{FB}$ and $T^{FB}$ are independent of $\alpha$.
\end{enumerate}
\end{corollary}
The proof of this result is highly similar to that of \Cref{cor:comp_stat_tmax} in \Cref{ss:comparative_FB} and is therefore omitted. In our opinion, the main feature to notice is that the quantity $e_{\rm min}^{FB}$ is decreasing with respect to $\varepsilon$, while the price is increasing. This means that in the first-best case, the insurer can force high-risk households to subscribe a small insured quantity at a higher price. Indeed, by knowing the type of the agents, the principal can design very specific contract for each type of agents. In particular, since the reservation utility ${\rm EU}^\varnothing$ is decreasing with respect to $\varepsilon$, a high-risk type agent is willing to pay more than a low-risk type agent for the same quantity, and the principal therefore exploits this fact.

\medskip

Given the optimal contract $(e_{\rm min}^{FB}, T^{FB})$, the maximum profit obtained by the principal for each risk type $\varepsilon$ of agent is computed explicitly in \Cref{cor:pb_principal_FB}. In particular, this profit is surprisingly independent of the price $p_e$. In other words, the profit of the insurer is not impacted by the variations of the staple good's price. The most important fact is that the profit is increasing with respect to $\varepsilon$. In other words, a risky household is much more profitable for the insurer, which is consistent with the previous remarks. Indeed, in the first-best case, the riskier the household, the smaller the insured quantity but the higher the price. One may note that the profit is also increasing in $w_0$.

\subsection{First-best insurance against fuel poverty}\label{ss:fuel_poverty_FB}

With the motivation described in \Cref{ss:fuel_poverty}, we apply the results of the first-best case in the context of insurance against fuel poverty in France. We consider the same household as the one studied in \Cref{sec:fuel_poverty_TB}. Thanks to \Cref{prop:pb_principal_FB}, we can compute the optimal insured quantity and the insurance premium, as well as the principal's profit (blue curves on \Cref{fig:FB_QI}). In this case, an agent of a higher risk type pays a higher insurance premium for a smaller insured quantity than does an agent of a lower risk type, as we can see by combining the two upper graphs of \Cref{fig:FB_QI}. More precisely, the insured quantity varies from approximately $e_{\rm min}^{FB}=14,403$ to $12,653$ kWh, while the price, $T^{FB}$ increases from $0$ to $4,252$ \euro{}. Moreover, the middle graph of \Cref{fig:FB_QI} shows that the riskier the agent's type, the greater the difference between the insurance premium and the actuarial price is. Since the actuarial price also corresponds to the principal's cost, the greater this difference, the greater her profit is. Therefore, from the principal's point of view, the `more efficient' agents are those who are willing to pay more than the actuarial price (bottom graph of \Cref{fig:FB_QI}). Her average profit $\Pi^{FB}  = 1,035$ \euro{} is given by the integral of her profit per agent, assuming that the distribution of risk type is uniform.

\medskip

Therefore, in our framework, the efficient agents are those who are at risk of losing their income, which may seem counterintuitive.
Nevertheless, this result can be explained by the reservation utility we have chosen.
More precisely, one can compute the information rent, which is the difference between ${\rm EU}(\varepsilon,q,t_0)$ and ${\rm EU}^\varnothing (\varepsilon)$ for an agent of type $\varepsilon$. This information rent increases with $\varepsilon$ for every $q$ and $t_0$, which means exactly that the riskier the agent's type, the more interesting it is for him to buy the insurance.
In the first-best case, the principal knows the agent's type and can thus reduce his information rent to zero. This explains why an agent of a higher risk type is ready to pay more for a lower insured quantity.

\medskip

This section gives some intuitions for the third-best case, \textit{i.e.}, when the principal cannot differentiate among the agents by their type since they are unknown to her. Indeed, under adverse selection, if the principal offers the optimal first-best contract, an agent with a positive probability of losing income may lie and pretend to be of a lower risk type to pay less for a higher insured quantity. Therefore, the agents of a higher risk type would receive an information rent generated by the informational advantage they have over the principal.
Conversely, the agent of the lower risk type should have no information rent since he has no incentive to lie: if he pretends he is of a higher risk type, he pays a higher premium for a lower insured quantity. 
As in classical adverse selection problems, here the \textit{efficient} agents, which are those who have an incentive to lie, receive an \textit{information rent}, generated by the informational advantage they have over the principal.
In line with this reasoning, it is clear that under adverse selection, the first-best contract should not be optimal anymore, and the principal thus has to find a new optimal menu of contracts.

\begin{figure}[!ht]
\includegraphics[scale=0.5]{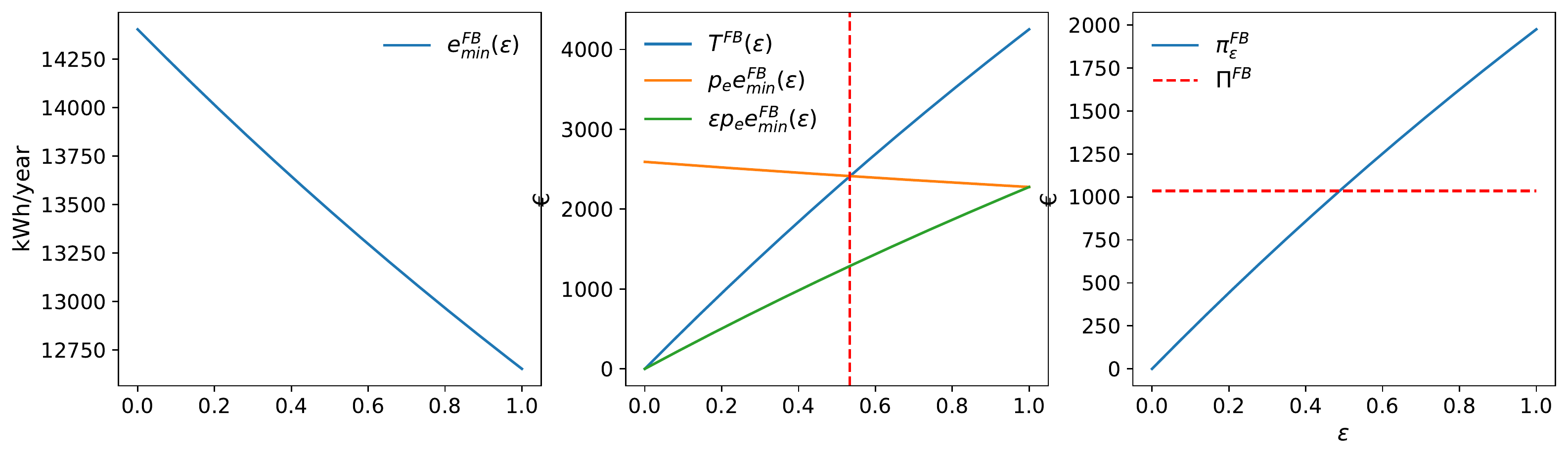}
\centering
\caption{Optimal insurance in the first-best case.}
{\footnotesize The blue curves represent, from left to right, the insured quantity, the insurance premium and the principal's profit, with respect to the type $\varepsilon$ of the agent. 
Middle: insurance price is compared to the future price of the quantity (orange curve) and to its actuarial price (green curve), which also corresponds to the principal's cost.
Right: red dotted line is the principal's average profit.}
\label{fig:FB_QI}
\end{figure}

\section{Technical results and proofs...}\label{app:tech_proof_SB}

\subsection{... to find a risk-revealing menu of contracts}\label{app:tech_proof_revealing} 

Before seeking a risk-revealing menu of contracts, we first prove that the Spence-Mirrlees condition is satisfied in our framework (see \Cref{lem:sm_condition}). This condition is important since it makes the incentive problem well behaved in the sense that only local incentive constraints need to be considered. Together with \Cref{lem:constant}, this allows us to establish \Cref{thm:IC_SB}, whose proof is reported below, after the two lemmas.
\begin{lemma}\label{lem:sm_condition}
The marginal rates of substitution between the in-kind support and the insurance price can be ranked in a monotonic way. More precisely,
\begin{align*}
        \dfrac{\partial}{\partial \varepsilon} \Bigg( \dfrac{\partial_{q} {\rm EU^{\rm Q}} \big(\varepsilon, q, t_0 \big)}{\partial_{t_0} {\rm EU^{\rm Q}} \big(\varepsilon, q, t_0 \big) } \Bigg) \leq 0.
    \end{align*}
\end{lemma}

\begin{proof}
Indeed, recalling that the expected utility of an agent of type $\varepsilon$ is given by \eqref{eq:exp_utility}, we have:
\begin{align*}
    \partial_{q} \textnormal{EU}^{\rm Q} \big(\varepsilon, q, t_0 \big) = &\
    \left\{
        \begin{array}{ll}
        \beta \varepsilon \alpha (1+\alpha) / (1 + q \alpha) & \mbox{ if } q < 1, \\
        \beta \varepsilon \alpha/q & \mbox{ if } q \geq 1,
        \end{array}
    \right.
    \text{ and } \; \partial_{t_0} \textnormal{EU}^{\rm Q} \big(\varepsilon, q, t_0 \big) = - \dfrac{1+\alpha}{1 - t_0},
\end{align*}
which leads to
\begin{align*}
        \dfrac{\partial}{\partial \varepsilon} \Bigg( \dfrac{\partial_{q} {\rm EU^{\rm Q}} \big(\varepsilon, q, t_0 \big)}{\partial_{t_0} {\rm EU^{\rm Q}} \big(\varepsilon, q, t_0 \big) } \Bigg)
        = 
    \left\{
        \begin{array}{ll}
        - \dfrac{\beta \alpha (1 - t_0)}{1 + q \alpha} & \mbox{ if } q < 1, \\
        - \dfrac{\beta \alpha (1 - t_0)}{q (1+\alpha)} & \mbox{ if } q \geq 1.
        \end{array}
    \right.
    \end{align*}
Both quotients are indeed non-positive since the price of the insurance should be at least smaller than the agents' income, which implies $t_0 < 1$, and all other quantities and prices are positive.
\end{proof}

\begin{lemma}\label{lem:constant}
Let $(q,t_0)$ be an admissible mechanism such that $q$ is non-decreasing and $t_0$ is given by $\eqref{eq:TQ_emin_eps}$ for some $c_q \geq 0$. If the function $q$ is constant on some interval contained in $[0,1]$, then the price $t_0$ is also constant on this interval.
\end{lemma}

\begin{proof}
Let us first assume that $q$ is constant on some interval $[x, y]$, where $y < \varepsilon_1$. For all $\varepsilon$ in this interval, we have in particular $q(\varepsilon) = q(x)$ and thus:
\begin{align*}
        t_0(\varepsilon) = &\ 1 - c_q \big( 1 + \alpha q (x) \big)^{- \beta x - \beta (\varepsilon - x)} \exp \bigg( \beta \int_0^{x} \ln \big( 1 + \alpha q(\epsilon) \big) \drm \epsilon
        + \beta \int_{x}^\varepsilon \ln \big( 1 + \alpha q(x) \big) \drm \epsilon \bigg) \\
        = &\ 1 - c_q \big( 1 + \alpha q (x) \big)^{- \beta x} \exp \bigg( \beta \int_0^{x} \ln \big( 1 + \alpha q(\epsilon) \big) \drm \epsilon \bigg).
    \end{align*}
Therefore, for all $\varepsilon \in [x, y]$, $t_0(\varepsilon) = t_0(x)$, \textit{i.e.}, $t_0$ is constant on this interval. The proof is highly similar for an interval $[x, y]$ such that $x \geq \varepsilon_1$.
Finally, if the interval $[x, y]$ contains $\varepsilon_1$, we necessarily have $q(\varepsilon) = 1$ for all $\varepsilon$ in the interval. By definition of $\varepsilon_1$, we actually have $x \geq \varepsilon_1$, and the problem is reduced to the previous case.
\end{proof}

\begin{proof}[\Cref{thm:IC_SB}]
$(i)$ We first prove that $q$ being non-decreasing with respect to $\varepsilon \in [0,1]$ and $t_0$ satisfying \eqref{eq:TQ_emin_eps} are necessary conditions for the admissible mechanism $(q, t_0)$ to satisfy the IC constraint on $[0,1]$. To prove this, we first fix $\varepsilon \in (0,1)$, such that $q(\varepsilon) \neq 1$, and focus the study on an agent of type $\varepsilon$. From \eqref{eq:exp_utility}, his expected utility if he chooses a contract $(q(\varepsilon'), t_0(\varepsilon'))$, for some $\varepsilon' \in [0,1]$, is as follows:
\begin{align}\label{eq:expected_utility_eps'}
    \textnormal{EU}^{\rm Q} \big(\varepsilon, q(\varepsilon'), t_0(\varepsilon') \big) = &\ \textnormal{EU}^\varnothing (\varepsilon) + (1+\alpha) \ln \big(1 - t_0(\varepsilon') \big) + \beta \varepsilon \widebar U \big( q(\varepsilon') \big).
\end{align}

The mechanism $(q, t_0)$ is incentive compatible if the agent chooses the contract designed for him to maximise his utility.
Therefore, the utility computed above must attain its maximum on $\varepsilon' = \varepsilon$. Since the mechanism is assumed to be regular enough, we can compute the first and second derivatives of the previous utility with respect to $\varepsilon'$. The first-order condition (FOC) states that the first derivative has to be equal to zero for $\varepsilon' = \varepsilon$. Since the derivative with respect to $\varepsilon' \in (0,1)$ of $\textnormal{EU}^{\rm Q} (\varepsilon, q(\varepsilon'), t_0(\varepsilon') )$ is given by:
\begin{align}\label{eq:der_expected_utility_eps'}
    \partial_{\varepsilon'} \textnormal{EU}^{\rm Q} \big(\varepsilon, q(\varepsilon'), t_0(\varepsilon') \big) = &\ - (1+\alpha) \dfrac{\partial_{\varepsilon'} t_0(\varepsilon')}{1 - t_0(\varepsilon')} + (1+ \alpha) \beta \varepsilon \times 
    \left\{
    \arraycolsep=1.4pt\def\arraystretch{1.8}
        \begin{array}{ll}
        \dfrac{\alpha \partial_{\varepsilon'} q(\varepsilon')}{1 + \alpha q(\varepsilon')} & \mbox{ if } q(\varepsilon') < 1, \\
        \dfrac{\alpha \partial_{\varepsilon'} q(\varepsilon')}{(1+ \alpha) q(\varepsilon')} & \mbox{ if } q(\varepsilon') > 1,
        \end{array}
    \right.
\end{align}
the FOC for the agent of type $\varepsilon$ is as follows:
\begin{align}\label{eq:FOC_revelation_principle}
    \partial_{\varepsilon} t_0(\varepsilon) = 
    \left\{
    \arraycolsep=1.4pt\def\arraystretch{2.}
        \begin{array}{ll}
        \beta \varepsilon \dfrac{\alpha \partial_{\varepsilon} q(\varepsilon)}{1 + \alpha q(\varepsilon)} \big( 1 - t_0(\varepsilon) \big) & \mbox{ if } q(\varepsilon) < 1, \\
        \beta \varepsilon \dfrac{\alpha \partial_{\varepsilon} q(\varepsilon)}{(1 + \alpha) q(\varepsilon)} \big( 1 - t_0(\varepsilon) \big) & \mbox{ if } q(\varepsilon) > 1.
        \end{array}
    \right.
\end{align}
Moreover, to check that $\varepsilon' = \varepsilon$ attains a local maximum, the second-order derivative has to be negative for $\varepsilon' = \varepsilon$ (second-order condition -- SOC), which gives:

\begin{align}\label{eq:SOC_revelation_principle}
    0 \geq - \dfrac{\partial^2_{\varepsilon} t_0(\varepsilon)}{1 - t_0(\varepsilon)}
    - \bigg( \dfrac{\partial_{\varepsilon} t_0(\varepsilon)}{1 - t_0(\varepsilon)} \bigg)^2 + \beta \varepsilon \times 
    \left\{
    \arraycolsep=1.4pt\def\arraystretch{2.}
        \begin{array}{ll}
        \dfrac{\alpha \partial^2_{\varepsilon} q(\varepsilon)}{1 + \alpha q(\varepsilon)}
        - \dfrac{\alpha^2 \big( \partial_{\varepsilon} q(\varepsilon) \big)^2}{\big( 1 + \alpha q(\varepsilon) \big)^2} & \mbox{ if } q(\varepsilon) < 1, \\
        \dfrac{\alpha \partial^2_{\varepsilon} q(\varepsilon)}{ (1 + \alpha) q(\varepsilon)}
        - \dfrac{\alpha  \big( \partial_{\varepsilon} q(\varepsilon) \big)^2}{ (1 + \alpha) \big( q(\varepsilon) \big)^2} & \mbox{ if } q(\varepsilon) > 1.
        \end{array}
    \right.
\end{align}
The mechanism $(q,t_0)$ must be revealing for every type of agent, which implies that the previous FOC and SOC have to be true at least for all $\varepsilon \in (0,1)$ such that $q(\varepsilon) \neq 1$.

\medskip

On the one hand, by differentiating \eqref{eq:FOC_revelation_principle}, we prove that $t_0$ should satisfy:
\begin{align*}
    \partial^2_{\varepsilon} t_0(\varepsilon) = 
    \left\{
    \arraycolsep=1.4pt\def\arraystretch{2.}
        \begin{array}{ll}
        \dfrac{\beta \alpha \big( 1 - t_0(\varepsilon) \big)}{\big( 1 + \alpha q(\varepsilon) \big)^2} \Big(
        \big( \partial_{\varepsilon} q(\varepsilon) + \varepsilon \partial^2_{\varepsilon} q(\varepsilon) \big) (1 + \alpha q(\varepsilon))
        - \alpha \varepsilon \big(\partial_{\varepsilon} q(\varepsilon) \big)^2 (1+\beta \varepsilon)
        \Big)
        & \mbox{ if } q(\varepsilon) < 1, \\
        \dfrac{\beta \alpha \big( 1 - t_0(\varepsilon) \big)}{(1 + \alpha) \big( q(\varepsilon) \big)^2} \bigg(
        \big( \partial_{\varepsilon} q(\varepsilon) + \varepsilon \partial^2_{\varepsilon} q(\varepsilon) \big) q(\varepsilon) 
        - \varepsilon \big( \partial_{\varepsilon} q(\varepsilon) \big)^2
        \bigg( 1 + \dfrac{\alpha }{1 + \alpha} \beta \varepsilon \bigg)  
        \bigg) 
        & \mbox{ if } q(\varepsilon) > 1.
        \end{array}
    \right.
\end{align*}
Replacing the first and second derivatives of $t_0$ with their values computed above, we obtain that the SOC \eqref{eq:SOC_revelation_principle} is equivalent in both cases to $\partial_{\varepsilon} q (\varepsilon) \geq 0$. Therefore, $q$ is non-decreasing before attaining $1$ and non-decreasing thereafter. By continuity of the function $q$, it can either cross the constant line equal to $1$ only once or be equal to $1$ over an interval. In both cases, it implies that the function $q$ is non-decreasing on $(0,1)$.
We can thus denote by $[\varepsilon_1, \varepsilon_2]$ the interval on which $q$ is constant equal to 1, with the convention that this interval is reduced to $\{ \varepsilon_1\}$ if there exists only one point where $q$ is equal to $1$, $\varepsilon_1 = \varepsilon_2 = 1$ if $q$ is always strictly less than $1$ and $\varepsilon_1 = \varepsilon_2 = 0$ if $q$ is always strictly greater than $1$.

\medskip

On the other hand, by solving \eqref{eq:FOC_revelation_principle} when $q(\varepsilon) < 1$, \textit{i.e.}, $\varepsilon \in (0, \varepsilon_1)$, we obtain
\begin{align*}
    t_0(\varepsilon)
    = &\ 1 - c_q \big(1 + \alpha q(\varepsilon)\big)^{- \beta \varepsilon}
    \erm^{\beta Q_0 ( \varepsilon )},
\end{align*}
for some constant $c_q \in \R$, using the notation defined in \eqref{eq:Q_0}. This proves the first form in \eqref{eq:TQ_emin_eps}. Moreover, solving the second part of \eqref{eq:FOC_revelation_principle}, \textit{i.e.}, for $\varepsilon \in (\varepsilon_2,1)$, leads to:
\begin{align*}
    t_0(\varepsilon)
    = &\ 1 - \widetilde c_q
    (q(\varepsilon))^{- \beta \varepsilon \alpha/(1 + \alpha) }
     \exp \bigg( \dfrac{\beta \alpha}{1 + \alpha} \int_{\varepsilon_2}^\varepsilon \ln (q(\epsilon))  \drm \epsilon \bigg),
\end{align*}
for some $\widetilde c_q \in \R$. The two previous forms are valid on $[0, \varepsilon_1]$ and $[\varepsilon_2, 1]$ by the assumed continuity of $t_0$. If $\varepsilon_1 = \varepsilon_2 \in (0,1)$, the price is continuous at this point if and only if
\begin{align}\label{eq:cqtilde}
    \widetilde c_q = c_q \big(1 + \alpha \big)^{- \beta \varepsilon_1}
    \erm^{ \beta Q_0(\varepsilon_1)},
\end{align}
and we thus obtain the second form in \eqref{eq:TQ_emin_eps}. With this setting, if $\varepsilon_1 = \varepsilon_2 = 0$, we obtain $\widetilde c_q = c_q$, and the price is given by \eqref{eq:TQ_emin_eps} for all $\varepsilon \in (0,1)$, to within the constant $c_q$. Similar reasoning applies if $\varepsilon_1 = \varepsilon_2 = 1$. It remains to us to deal with the case where $q$ is constant on the interval $[\varepsilon_1, \varepsilon_2]$, not reduced to a singleton. To address this case, let us consider an agent of type $\varepsilon \in (\varepsilon_1, \varepsilon_2)$. For the mechanism to be revealing for him, his expected utility $\textnormal{EU}^{\rm Q} (\varepsilon, q(\varepsilon'), t_0(\varepsilon') )$ must at least reach a local maximum in $\varepsilon' = \varepsilon$. His utility for any $\varepsilon' \in (\varepsilon_1, \varepsilon_2)$ is as follows, since $q$ is constant equal to $1$ on this interval:
\begin{align*}
    \textnormal{EU}^{\rm Q} \big(\varepsilon, q(\varepsilon'), t_0(\varepsilon') \big) = &\ \textnormal{EU}^\varnothing (\varepsilon) + (1+\alpha) \ln \big(1 - t_0(\varepsilon') \big) + \beta \varepsilon (1+\alpha) \ln \big(1 + \alpha \big).
\end{align*}
Therefore, $\varepsilon$ is a maximum point on $(\varepsilon_1, \varepsilon_2)$ if for any $\varepsilon'$ on this interval, $t_0(\varepsilon) \leq t_0(\varepsilon')$. Conversely, the mechanism is revealing for the agent of type $\varepsilon'$ at least if $t_0(\varepsilon) \geq t_0(\varepsilon')$. This naturally implies that $t_0$ is also constant on $(\varepsilon_1, \varepsilon_2)$. In particular, by continuity of the price, we should have $t_0(\varepsilon_1) = t_0(\varepsilon_2)$, which also implies \eqref{eq:cqtilde}. Finally, we obtain for all $\varepsilon \in [\varepsilon_2, 1]$:
\begin{align*}
    t_0(\varepsilon)
    = &\ 1 - c_q \big(1 + \alpha \big)^{- \beta \varepsilon_1}
    (q(\varepsilon))^{- \beta \varepsilon \alpha/(1 + \alpha) }
    \erm^{\beta Q_0(\varepsilon_1)}
     \exp \bigg( \dfrac{\beta \alpha}{1 + \alpha} \int_{\varepsilon_2}^\varepsilon \ln (q(\epsilon))  \drm \epsilon \bigg) \\
     = &\ 1 - c_q \big(1 + \alpha \big)^{- \beta \varepsilon_1}
    (q(\varepsilon))^{- \beta \varepsilon \alpha/(1 + \alpha) }
    \erm^{\beta Q_0(\varepsilon)},
    \end{align*}
where the second equality is implied by the fact that $q$ is constant equal to $1$ on $[\varepsilon_1, \varepsilon_2]$.
Therefore, the form \eqref{eq:TQ_emin_eps} is proven to be true in any case.
Finally, since the mechanism has to be admissible in the sense of \Cref{def:regularity}, we should have $t_0(\varepsilon) < 1$ for all $\varepsilon \in [0,1]$, which implies $c_q > 0$.
We therefore have shown that $q$ being non-decreasing with respect to $\varepsilon \in [0,1]$ and $t_0$ satisfying \eqref{eq:TQ_emin_eps} are \textbf{necessary conditions} for the menu of contracts to satisfy the IC constraint on $[0,1]$.

\medskip

$(ii)$ It remains to us to prove that these conditions are \textbf{sufficient}. To this end, we recall that the expected utility of an agent of type $\varepsilon$ who chooses a contract $(q(\varepsilon'), t_0(\varepsilon'))$ is given by \eqref{eq:expected_utility_eps'}. In particular, its derivative with respect to $\varepsilon'$ for $\varepsilon' \in (0,1)$ such that $q(\varepsilon') \neq 1$ is given by \eqref{eq:der_expected_utility_eps'}.
Since $t_0$ satisfies \eqref{eq:FOC_revelation_principle} in particular in $\varepsilon'$, we obtain
\begin{align*}
    \partial_{\varepsilon'} \textnormal{EU}^{\rm Q} \big(\varepsilon, q(\varepsilon'), t_0(\varepsilon') \big) = &\ 
    \left\{
    \arraycolsep=1.4pt\def\arraystretch{2.}
        \begin{array}{ll}
        (1+\alpha) \dfrac{\beta \alpha \partial_{\varepsilon'} q(\varepsilon')}{1 + \alpha q(\varepsilon')} \big(\varepsilon - \varepsilon' \big)  & \mbox{ if } q(\varepsilon') < 1, \\
        \beta \dfrac{\alpha \partial_{\varepsilon'} q(\varepsilon')}{q(\varepsilon')} \big( \varepsilon - \varepsilon' \big) & \mbox{ if } q(\varepsilon') > 1.
        \end{array}
    \right.
\end{align*}
Moreover, if we consider without loss of generality that $q$ is constant equal to $1$ on some interval $[\varepsilon_1, \varepsilon_2]$ and take $\varepsilon \in (\varepsilon_1, \varepsilon_2)$, then, by \Cref{lem:constant}, for any neighbourhood of $\varepsilon$ contained in $[\varepsilon_1, \varepsilon_2]$, the price given by \eqref{eq:TQ_emin_eps} is also constant. Therefore, the expected utility ${\rm EU}^{\rm Q} (\varepsilon, q(\varepsilon'), t_0(\varepsilon'))$ is, in fact, also differentiable on this neighbourhood, and its derivative is equal to zero. In summary, the following values are obtained for the derivative of the expected utility:
\begin{align*}
    \partial_{\varepsilon'} \textnormal{EU}^{\rm Q} \big(\varepsilon, q(\varepsilon'), t_0(\varepsilon') \big) = &\ 
    \left\{
    \arraycolsep=1.4pt\def\arraystretch{1.2}
        \begin{array}{ll}
        \beta (1+\alpha) \dfrac{\alpha \partial_{\varepsilon'} q(\varepsilon')}{1 + \alpha q(\varepsilon')} \big(\varepsilon - \varepsilon' \big)  & \mbox{ if } 0 < \varepsilon' < \varepsilon_1, \\
        0  & \mbox{ if } \varepsilon' \in (\varepsilon_1, \varepsilon_2), \\
        \beta  \alpha \dfrac{\partial_{\varepsilon'} q(\varepsilon')}{q(\varepsilon')} \big( \varepsilon - \varepsilon' \big) & \mbox{ if } \varepsilon_2 < \varepsilon' < 1.
        \end{array}
    \right.
\end{align*}
We first check that the contract is revealing for the interior types of agents, \textit{i.e.}, where the previous derivative is defined. It suffices to remark that the expected utility of an agent of type $\varepsilon$ is non-decreasing for $\varepsilon' \leq \varepsilon$ and non-increasing after, which proves, by continuity of the utility, that $\varepsilon' = \varepsilon$ is a maximiser.
If the agent's type is $\varepsilon = 0$, his continuous utility is non-increasing with $\varepsilon' \in (0,1)$, and a maximum is attained for $\varepsilon' = 0$. Similar reasoning can be applied if the agent's type is $\varepsilon = 1$, and therefore, the contract is revealing for the extreme types. For $\varepsilon = \varepsilon_1$ (resp. $\varepsilon = \varepsilon_2$), the utility is non-decreasing before $\varepsilon$, constant on $(\varepsilon_1, \varepsilon_2)$, and non-increasing after. Therefore, the (continuous) utility is constant and maximal on the interval $[\varepsilon_1, \varepsilon_2]$; in particular, the maximum is also attained at $\varepsilon_1$ (resp. $\varepsilon_2$). Therefore, the conditions stated in the proposition are \textbf{sufficient} for the mechanism to satisfy the IC constraint for all $\varepsilon \in [0,1]$.
\end{proof}

\Cref{thm:IC_SB} thus provides a characterisation of a mechanism $(q, t_0)$ satisfying the IC constraint for all types of agents. However, the real menu of contracts offered by the principal must be composed of quantities $e_{\rm min}$ and a price $T$ associated with each quantity, independent of the type of agent, which is not observed by the principal. Therefore, in the end, we have to obtain a price $T$ that is only a function of $e_{\rm min}$, not also a function of $\varepsilon$.
Nevertheless, \Cref{lem:constant} states that when the function $q$ is constant, the associated price $t_0$ is necessarily constant as well. Together with the fact that the function $q$ is non-decreasing, this naturally implies that if two different types of agents choose the same quantity, they pay the same price. This result therefore prevents the contract resulting from a revealing mechanism from depending on the type of agent.

\medskip

To precisely define the menu of contracts associated with an admissible revealing mechanism, let us fix an interval $I \subset \R_+$ and define $\widetilde I := \{k \in \R_+ \; \text{s.t.} \; \alpha k \underline \omega w_0/p_e \in I \}$. For a function $f$ that is non decreasing on $[0,1]$, taking values in $\widetilde I$, its generalised inverse for all $k \in \widetilde I$ is defined by:
\begin{align}\label{eq:gen_inverse}
    f^{-1} (k) = \inf \{ \varepsilon \in [0,1] \; \text{ such that } \; f(\varepsilon) = k \}.
\end{align}
The following corollary allows us to characterise a sufficiently smooth admissible menu of revealing contracts. The proof of this result is highly similar to that of \Cref{thm:IC_SB}.

\begin{corollary}\label{cor:contract_emin}
An admissible menu of contracts $(e_{\rm min}, T)$ for $e_{\rm min} \in I$ is associated with an admissible revealing mechanism if and only if there exists a non-decreasing continuous function $q$ with values in $\widetilde I$ and continuous second derivatives except where it is equal to $1$, such that the price $T$ for a quantity $e_{\rm min} = \alpha k \underline\omega w_0/p_e$ is given for some $c_q \geq 0$ by:
\begin{align}\label{eq:TQ_emin}
T(k) = w_0 - c_q w_0 \erm^{\beta Q_0( q^{-1} (k))} \times 
\left\{
\arraycolsep=1.6pt\def\arraystretch{1.3}
    \begin{array}{ll}
    \big( 1 + \alpha k \big)^{- \beta q^{-1} (k)},
    & \mbox{ if } k < 1, \\
    \big( 1 + \alpha \big)^{- \beta q^{-1} (1)} k^{- \beta q^{-1} (k) \alpha/(1+\alpha)},
    & \mbox{ if } k \geq 1
    \end{array}
\right.
\end{align}
for $k \in \widetilde I$ and where $q^{-1}$ is the generalised inverse of $q$, as defined in \eqref{eq:gen_inverse}.
\end{corollary}

\begin{proof}
$(i)$ To prove that it is a necessary condition, let us fix an admissible menu of contracts $(e_{\rm min}, T)$ and an associated admissible revealing mechanism $(q,t_0)$. Since the mechanism $(q,t_0)$ is admissible and satisfies the IC constraint, by \Cref{thm:IC_SB}, we obtain that $q$ is non-decreasing, and the price function $t_0$ is given by \eqref{eq:TQ_emin_eps} with a constant $ c_q > 0$.
Moreover, by the previous discussion on admissible contracts, $T$ should be independent of $\varepsilon$ and thus constant when $q$ is constant, which is true by \Cref{lem:constant}. Hence, we can write the price $t_0$ given by \eqref{eq:TQ_emin_eps} in $\varepsilon = q^{-1} (k)$, where $k := p_e e_{\rm min} / (\alpha \underline\omega w_0)$ and $q^{-1}$ is the generalised inverse of $q$. Moreover, noticing that $\varepsilon < \varepsilon_1$ is equivalent to $k < q^{-1}(\varepsilon_1) = 1$, and conversely, if $\varepsilon \geq \varepsilon_1$ then $k \geq 1$, we obtain that the price of a quantity $e_{\rm min} := \alpha k \underline\omega w_0$ is given by $T(k) = w_0 t_0 (q^{-1} (k))$, which is \eqref{eq:TQ_emin}. 

\medskip

$(ii)$ To prove the equivalence, let us consider an admissible menu of contracts $(e_{\rm min}, T)$, where $T$ is given by \eqref{eq:TQ_emin}, and assume that the function $q$ has the right properties. First, we can show that given this menu of contracts, the optimal quantity chosen by an agent of type $\varepsilon$ is $e_{\rm min} = \alpha k \underline\omega w_0/p_e$, where $k = q(\varepsilon)$. Indeed, by computing the derivative of his utility given by \eqref{eq:exp_utility} with respect to the normalised quantity $k$,
we obtain the following FOC:
\begin{align*}
0 = -  \dfrac{\partial_k t_0 (k)}{1-t_0 (k)} + \beta \varepsilon \times
    \left\{
    \arraycolsep=1.4pt\def\arraystretch{1.5}
        \begin{array}{ll}
        \dfrac{\alpha}{1 + \alpha k} & \mbox{ if } k < 1, \\
        \dfrac{\alpha}{(1+\alpha)k}  & \mbox{ if } k > 1.
        \end{array}
    \right.
\end{align*}
Since the derivative of $t_0$ with respect to $k$ satisfies:
\begin{align*}
\partial_k t_0 (k) = 
\left\{
\arraycolsep=1.4pt\def\arraystretch{2.}
    \begin{array}{ll}
    \dfrac{\beta  \alpha q^{-1}(k) }{1 + \alpha k} \big( 1-t_0(k) \big) & \mbox{ if } k < 1, \\
    \dfrac{\beta  \alpha q^{-1}(k)}{(1+\alpha) k} \big( 1-t_0(k) \big)  & \mbox{ if } k > 1,
    \end{array}
\right.
\end{align*}
the previous FOC is equivalent to $k = q(\varepsilon)$. By continuity of $q$, the result is extendable to $k=1$. It remains to check the following SOC:
\begin{align*}
0 \geq -  \dfrac{\partial^2_k t_0 (k) \big( 1-t_0 (k) \big) + \big(\partial_k t_0 (k)\big)^2 }{\big( 1-t_0 (k) \big)^2} - \beta \varepsilon \times
    \left\{
    \arraycolsep=1.4pt\def\arraystretch{2.}
        \begin{array}{ll}
        \dfrac{\alpha^2}{(1 + \alpha k)^2} & \mbox{ if } k < 1, \\
        \dfrac{\alpha}{(1+\alpha)k^2}  & \mbox{ if } k > 1.
        \end{array}
    \right.
\end{align*}
The second-order derivative of $T$ satisfies:
\begin{align*}
\partial^2_k t_0 (k) = 
\left\{
\arraycolsep=1.4pt\def\arraystretch{2.}
    \begin{array}{ll}
    \dfrac{\beta  \alpha \big( 1-t_0(k) \big)}{(1 + \alpha k)^2} \Big( \partial_k q^{-1}(k) (1 + \alpha k) - \alpha q^{-1}(k)
    - \beta  \alpha \big(q^{-1}(k)\big)^2 \Big) & \mbox{ if } k < 1, \\
    \dfrac{\beta  \alpha \big( 1-t_0(k) \big) }{(1+\alpha)^2 k^2} \Big(
    \big( \partial_k q^{-1}(k) k - q^{-1}(k) \big) (1+\alpha)
    - \beta \alpha  \big(q^{-1}(k) \big)^2
    \Big)  & \mbox{ if } k > 1.
    \end{array}
\right.
\end{align*}
Thus, the SOC is equivalent to:
\begin{align*}
\left\{
\arraycolsep=1.4pt\def\arraystretch{1.5}
\begin{array}{ll}
    \partial_k q^{-1}(k) (1 + \alpha k) + \alpha \big( \varepsilon  - q^{-1}(k) \big)  \geq 0  & \mbox{ if } k < 1, \\
    \partial_k q^{-1}(k) k - q^{-1}(k) + \varepsilon \geq 0  & \mbox{ if } k > 1.
\end{array}
\right.
\end{align*}
In $k = q(\varepsilon) \neq 1$, the SOC becomes in both cases $\partial_k q^{-1}(k) \geq 0$, which is true since $q$ is non-decreasing. By continuity of the utility, this result is also true for $k = 1$. Therefore, an agent of type $\varepsilon$ chooses the quantity $e_{\rm min} = \alpha q(\varepsilon) \underline\omega w_0/p_e$, which is an available quantity because $q$ takes values in $\widetilde I$ and thus $e_{\rm min} \in I$.
By computing the function $T(k)$ for $k = q(\varepsilon)$ and dividing it by $w_0$, we recover the function $t_0$ defined by \eqref{eq:TQ_emin_eps}, which associates with any $\varepsilon$ the price $w_0 t_0(\varepsilon)$ of the normalised quantity $k = q(\varepsilon)$.
Moreover, since the mechanism $(q,t_0)$ satisfies the assumptions to be admissible in the sense of \Cref{def:regularity}, by \Cref{thm:IC_SB}, the mechanism associated with the menu $(e_{\rm min}, T)$ is admissible and satisfies the IC constraint.
\end{proof}

\subsection{...to select the agents}\label{app:tech_proof_participation}

Given a mechanism $(q, t_0)$, we can write the informational rent of an agent of type $\varepsilon$ as a function of $\varepsilon$:
\begin{align*}
    \Delta \textnormal{EU}^{\rm Q} (\varepsilon) = &\ (1+\alpha) \ln \big(1 - t_0(\varepsilon) \big) + \beta \varepsilon \widebar U \big( q(\varepsilon) \big).
  \end{align*}
If $(q, t_0)$ is a revealing mechanism, we can use the FOC \eqref{eq:FOC_revelation_principle} to compute the derivative of the previous quantity:
\begin{align*}
    \partial_{\varepsilon} \Delta \textnormal{EU}^{\rm Q} (\varepsilon) = &\ \beta \widebar U \big( q (\varepsilon) \big) =
     \left\{
    \arraycolsep=1.4pt\def\arraystretch{1.5}
        \begin{array}{ll}
        \beta (1+\alpha) \ln \big(1 + \alpha q(\varepsilon) \big)
        & \mbox{ if } q(\varepsilon) < 1, \\
        \beta \alpha \ln \big( q(\varepsilon) \big) 
        + \beta (1+\alpha) \ln ( 1+ \alpha)
        & \mbox{ if } q(\varepsilon) > 1.
        \end{array}
    \right.
\end{align*}
This derivative is non-negative in both cases and implies that the information rent is non-decreasing. Therefore, if there exists $\underline \varepsilon \in [0,1]$ such that $\Delta \textnormal{EU}^{\rm Q} (\underline \varepsilon) \geq 0$, then for all $\varepsilon \in [\underline \varepsilon, 1]$, $\Delta \textnormal{EU}^{\rm Q} (\varepsilon) \geq 0$, which means that the participation constraint of agents of type $\varepsilon \in [\underline \varepsilon, 1]$ is satisfied. A more precise result is established in \Cref{prop:PC_for_RevContracts}, and its proof is reported below.

\begin{proof}[\Cref{prop:PC_for_RevContracts}]
We consider an admissible and incentive compatible mechanism $(q,t_0)$.
Applying \Cref{thm:IC_SB}, the price $t_0$ satisfies \eqref{eq:TQ_emin_eps}.
In the one hand, if $\varepsilon \in [0,1]$ is such that $q(\varepsilon) < 1$, \textit{i.e.}, $\varepsilon < \varepsilon_1$, the participation constraint of the agent of type $\varepsilon$ becomes:
\begin{align*}
    c_q \geq \exp \bigg(- \beta \int_0^\varepsilon \ln \big( 1 + \alpha q(\epsilon) \big) \drm \epsilon \bigg) = \underline c(\varepsilon).
\end{align*}
On the other hand, if $\varepsilon \geq \varepsilon_1$, the participation constraint is equivalent to $c_q 
    \geq (1+ \alpha)^{- \beta (\varepsilon - \varepsilon_1)} 
    \erm^{\beta Q_0(\varepsilon)}
      = \underline c(\varepsilon)$.
The participation constraint for an agent of type $\varepsilon \in [0,1]$ is thus equivalent in both cases to $c_q \geq \underline c (\varepsilon)$. We can then compute the derivative of $\underline c$ with respect to $\varepsilon$:
\begin{align*}
\underline c' (\varepsilon)
= - \beta \underline c(\varepsilon) \times 
\left\{
\arraycolsep=1.5pt\def\arraystretch{1.3}
    \begin{array}{ll}
    \ln ( 1 + \alpha q(\varepsilon)) 
    & \mbox{ if } \varepsilon < \varepsilon_1, \\
    \ln (1+ \alpha)
    + \dfrac{\alpha}{1+\alpha} \ln \big( q(\varepsilon) \big)
    & \mbox{ if } \varepsilon > \varepsilon_1.
    \end{array}
\right.
\end{align*}
Since $q(\varepsilon) \geq 0$ for all $\varepsilon \in [0,\varepsilon_1)$ and $q(\varepsilon) \geq 1$ for all $\varepsilon \in (\varepsilon_1,1]$, we obtain that the derivative of $\underline c$ is negative in both cases. Since the function $\underline c$ is continuous on $[0,1]$ (in particular in $\varepsilon_1$), the function is non-increasing on $[0,1]$. Moreover, by definition of $\underline \varepsilon$ and continuity of $\underline c$, $\underline c (\underline \varepsilon) = c_q$. Thus, for any $\varepsilon \in [\underline \varepsilon, 1]$, we have $\underline c(\varepsilon) \leq \underline c(\underline \varepsilon) = c_q$, and thus the participation constraint of the agent of type $\varepsilon$ is satisfied. Conversely, for any $\varepsilon \in [0, \underline \varepsilon)$, we have $\underline c(\varepsilon) > \underline c(\underline \varepsilon) = c_q$, which means that the participation constraint is not satisfied.
\end{proof}

\Cref{prop:PC_for_RevContracts} thus states that only the agents of a sufficiently risky type are selected by the principal. This result is entirely implied by the fact that the reservation utility of an agent depends on his type and only occurs in principal-agent problems with countervailing incentives. Indeed, the following remark shows that if a constant reservation utility had been chosen, the selected agents would have been those of a sufficiently low risk type. Nevertheless, in our framework, it makes little sense to consider that the reservation utility is constant, regardless of the agent's type.

\begin{remark}\label{rk:constant_reservation}
If the agents' reservation utility is assumed to be a constant $R_0$, the participation constraint for an agent of type $\varepsilon$ becomes
$ \textnormal{EU}^{\rm Q} (\varepsilon) \geq R_0$, where $\textnormal{EU}^{\rm Q} (\varepsilon)$ is defined by \eqref{eq:exp_utility} for a revealing contract $(q(\varepsilon), t_0(\varepsilon))$.
By computing the derivative of $\textnormal{EU}^{\rm Q} (\varepsilon)$ with respect to $\varepsilon$ for a menu of revealing contracts, using \textnormal{FOC \eqref{eq:FOC_revelation_principle}}, we obtain:
\begin{align*}
\partial_{\varepsilon} \textnormal{EU}^{\rm Q} (\varepsilon) 
= &\ (1+\alpha) \beta \ln \big( \underline \omega / \widebar \omega \big) + \beta \widebar U \big( q(\varepsilon) \big).
\end{align*}
Under the assumption\footnote{This is the case in the application considered throughout this paper, since $\alpha = 0.08$, $\underline\omega = 0.4$ and $\widebar\omega= 1$.} that $(1+\alpha) \underline\omega \leq \widebar\omega$, the information rent $\textnormal{EU}^{\rm Q} (\varepsilon) - R_0$ is decreasing for all $\varepsilon \in [0,1]$ such that $q(\varepsilon) \in \big(1, ( \widebar\omega/(\underline\omega ( 1+ \alpha)))^{(1+\alpha)/\alpha} \big]$. Thus, in this case, if there exists $\widebar \varepsilon \in [0,1]$ such that $ \textnormal{EU}^{\rm Q} (\widebar \varepsilon) \geq R_0$, then the participation constraint of agents of type $\varepsilon \in [0, \widebar \varepsilon]$ is satisfied.
\end{remark}

\subsection{...to solve the principal's problem}\label{app:tech_proof_principal}

\begin{corollary}\label{cor:pb_principal_SB}
Let $\underline \varepsilon \in [0,1]$ and $\underline q \in \R_+$. If there is a solution $Q \in \Qc(\underline \varepsilon)$ to \textnormal{ODE \eqref{eq:ODE}}, the principal's profit defined by \eqref{eq:principal_pb_simplify} is equal to:
\begin{align*}
    \Pi (\underline \varepsilon) = &\ w_0 (1 - \underline \varepsilon) + \dfrac{1}{2} \underline \omega w_0 \big( ( \varepsilon_1 \vee \underline \varepsilon)^2 - \underline \varepsilon^2 \big) 
    - w_0 F_1(Q) - w_0 F_2(Q),
\end{align*}
where $F_1$ and $F_2$ are, respectively, defined as follows:
\begin{subequations}\label{eq:F}
\begin{align}
F_1(Q) &:= \int_{\underline \varepsilon}^{\varepsilon_1 \vee \underline \varepsilon} \Big(
    \erm^{\beta (Q(\varepsilon) - \varepsilon Q'(\varepsilon))}
    + \varepsilon \underline\omega \erm^{Q' ( \varepsilon )} 
    \Big) \drm \varepsilon \label{eq:F_1} \\
F_2 (Q) &:= \int_{\varepsilon_1 \vee \underline \varepsilon}^1 \Big(
    \big( 1 + \alpha \big)^{- \beta (\varepsilon_1 \vee \underline \varepsilon)} 
    \erm^{\beta (Q(\varepsilon)- \varepsilon  Q' ( \varepsilon ))} 
    + \varepsilon \alpha \erm^{\frac{1+\alpha}{\alpha} Q' ( \varepsilon )} \underline\omega \Big) \drm \varepsilon. \label{eq:F_2}
\end{align}
\end{subequations}
\end{corollary}

\begin{proof}[\Cref{thm:sol_principal_ODE} and \Cref{cor:pb_principal_SB}]

Let us fix a mechanism $(q,t_0) \in \Cc^{\rm Q}(\underline \varepsilon)$. This mechanism satisfies the assumption of \Cref{thm:IC_SB}, and the price $t_0$ is therefore given by \eqref{eq:TQ_emin_eps}. Moreover, since this mechanism is assumed to be in $\Cc^{\rm Q}(\underline \varepsilon)$, the participation constraint has to be satisfied only for all $\varepsilon \in [\underline \varepsilon, 1]$, which implies by \Cref{prop:PC_for_RevContracts} that the constant $c_q$ in the price is given by $c_q = \underline c (\underline \varepsilon)$. We thus obtain that, if $\underline \varepsilon \in [0, \varepsilon_1)$,
the price is given by:
\begin{align*}
t_0(\varepsilon) = 1 - \erm^{\beta (Q_0(\varepsilon) - Q_0(\underline \varepsilon))} \times 
\left\{
\arraycolsep=1.6pt\def\arraystretch{1.4}
    \begin{array}{ll}
    \big( 1 + \alpha q (\varepsilon) \big)^{- \beta \varepsilon},
    & \mbox{ if } \varepsilon \in [\underline \varepsilon, \varepsilon_1), \\
    \big( 1 + \alpha \big)^{- \beta \varepsilon_1} \big( q (\varepsilon) \big)^{- \beta \varepsilon \alpha/(1+\alpha)},
    & \mbox{ if } \varepsilon \in [\varepsilon_1, 1].
    \end{array}
\right.
\text{ for all } \; \varepsilon \in [\underline \varepsilon, 1].
\end{align*}
Similarly, if $\underline \varepsilon \in [\varepsilon_1, 1]$, the price is given by:
\begin{align*}
t_0(\varepsilon) = 1 - \erm^{\beta (Q_0(\varepsilon) - Q_0(\underline \varepsilon))}
(1+ \alpha)^{- \beta \underline \varepsilon}  \big( q (\varepsilon) \big)^{- \beta \varepsilon \alpha/(1+\alpha)}, \; \text{ for all } \; \varepsilon \in [\underline \varepsilon, 1].
\end{align*}
To reconcile the two cases, we denote by $Q$ the following function for all $\varepsilon \in [\underline \varepsilon, 1]$:
\begin{align}
Q ( \varepsilon )
:= 
\left\{
\arraycolsep=1.5pt\def\arraystretch{1.8}
    \begin{array}{ll}
    \displaystyle\int_{\underline \varepsilon}^\varepsilon \ln \big( 1 + \alpha q(\epsilon) \big) \drm \epsilon
    & \mbox{ if } \varepsilon \in [\underline \varepsilon, \varepsilon_1 \vee \underline \varepsilon) \\
    \displaystyle\int_{\underline \varepsilon}^{\varepsilon_1 \vee \underline \varepsilon} \ln \big( 1 + \alpha q(\epsilon) \big) \drm \epsilon
    + \dfrac{\alpha}{1+\alpha} \displaystyle \int_{\varepsilon_1 \vee \underline \varepsilon}^\varepsilon \ln \big( q(\epsilon) \big) \drm \epsilon
    & \mbox{ if } \varepsilon \in [\varepsilon_1 \vee \underline \varepsilon, 1].
    \end{array}
\right.
\end{align}
Since $(q,t_0)$ is an admissible revealing mechanism, $q$ is continuous on $[0,1]$ and $\Cc^2$ on $(0,1)$ except where it is equal to $1$, and by \Cref{thm:IC_SB}, $q$ is a non-decreasing function. This naturally implies that the function $Q$ satisfies the right properties to be in $\Qc(\underline \varepsilon)$. Moreover, thanks to the definition of the function $Q$, we can write $q$ as a function of $Q'$:
\begin{align*}
q(\varepsilon) = 
\left\{
\arraycolsep=1.5pt\def\arraystretch{1.4}
    \begin{array}{ll}
    \dfrac{1}{\alpha} \big( \erm^{Q' ( \varepsilon )} -1 \big)
    & \mbox{ if } \varepsilon \in [\underline \varepsilon,  \varepsilon_1 \vee \underline \varepsilon) \\
    \erm^{\frac{1+\alpha}{\alpha} Q' ( \varepsilon )}
    & \mbox{ if } \varepsilon \in [\varepsilon_1 \vee \underline \varepsilon, 1].
    \end{array}
\right.
\end{align*}
Therefore, the price $t_0$ can be written as follows for all $\varepsilon \in [\underline \varepsilon, 1]$:
\begin{align*}
t_0(\varepsilon) = 
\left\{
\arraycolsep=1.4pt\def\arraystretch{1.2}
\begin{array}{ll} 
    1 - \erm^{\beta (Q(\varepsilon) - \varepsilon Q'(\varepsilon))}
    & \mbox{ if } \varepsilon \in [\underline \varepsilon, \varepsilon_1 \vee \underline \varepsilon), \\
    1 - \big( 1 + \alpha \big)^{- \beta (\varepsilon_1 \vee \underline \varepsilon)} 
    \erm^{\beta (Q(\varepsilon)- \varepsilon  Q' ( \varepsilon ))} 
    & \mbox{ if } \varepsilon \in [\varepsilon_1 \vee \underline \varepsilon, 1].
\end{array}
\right.
\end{align*}
Moreover, optimising on admissible revealing mechanisms $(q,t_0) \in \Cc^{\rm Q}(\underline \varepsilon)$ is thus equivalent to optimising on $Q \in \Qc(\underline \varepsilon)$, and the principal's problem for $\underline \varepsilon \in [0,1]$ fixed is thus given by:
\begin{align}\label{eq:profit_moche}
    \Pi (\underline \varepsilon) 
    = &\ w_0(1 - \underline \varepsilon) + \dfrac{1}{2} \underline\omega w_0 \big( ( \varepsilon_1 \vee \underline \varepsilon)^2 - \underline \varepsilon^2 \big) 
    - w_0 \inf_{Q \in \Qc(\underline \varepsilon)} \Bigg\{ 
    \int_{\underline \varepsilon}^{\varepsilon_1 \vee \underline \varepsilon} \Big(
    \erm^{\beta (Q(\varepsilon) - \varepsilon Q'(\varepsilon))} 
    + \varepsilon \underline\omega \erm^{Q' ( \varepsilon )} 
    \Big) \drm \varepsilon \nonumber \\
    &+ w_0 \int_{\varepsilon_1 \vee \underline \varepsilon}^1 \Big(
    \big( 1 + \alpha \big)^{- \beta (\varepsilon_1 \vee \underline \varepsilon)} 
    \erm^{\beta (Q(\varepsilon)- \varepsilon  Q' ( \varepsilon ))} 
    + \varepsilon \alpha \erm^{\frac{1+\alpha}{\alpha} Q' ( \varepsilon )} \underline\omega \Big) \drm \varepsilon 
    \Bigg\}.
\end{align}
This problem is relatively standard in the field of calculus of variation. Given the form of the previous optimisation, we study the optimal function $Q$ separately on $(\underline \varepsilon, \varepsilon_1 \vee \underline \varepsilon)$ and on $(\varepsilon_1 \vee \underline \varepsilon,1)$.

\medskip

With this in mind, we first study the problem on $(\underline \varepsilon, \varepsilon_1 \vee \underline \varepsilon)$.
Let $R$ be an arbitrary function that has at least one derivative and vanishes at the endpoints $\underline \varepsilon$ and $\varepsilon_1 \vee \underline \varepsilon$. For any $\eta \in \R$, we denote $g_1(\eta) := F_1(Q + \eta R)$, where $F_1$ is defined by \eqref{eq:F_1}. We can compute the derivative of $g_1$ with respect to $\eta$:
\begin{align*}
    g_1 '(\eta) = \int_{\underline \varepsilon}^{\varepsilon_1 \vee \underline \varepsilon} \Big(
    \beta \big( R(\varepsilon)  - \varepsilon R'(\varepsilon) \big)
    \erm^{\beta (Q(\varepsilon) + \eta R(\varepsilon) - \varepsilon Q'(\varepsilon) - \varepsilon \eta R'(\varepsilon))} 
    + \varepsilon \underline \omega R'(\varepsilon) \erm^{Q'(\varepsilon) + \eta R'(\varepsilon)} 
    \Big) \drm \varepsilon.    
\end{align*}
The Gâteaux differential of $F_1$ with respect to $Q$ in the direction $R$ denoted by ${\rm D} F_1(Q) (R)$ is given by $g_1 '(0)$:
\begin{align*}
    {\rm D} F_1(Q) (R) &= \int_{\underline \varepsilon}^{\varepsilon_1 \vee \underline \varepsilon} \Big(
    \beta \big( R(\varepsilon) - \varepsilon R'(\varepsilon) \big)
    \erm^{\beta (Q(\varepsilon) - \varepsilon Q'(\varepsilon))} 
    + \varepsilon \underline \omega R'(\varepsilon) \erm^{Q'(\varepsilon)} 
    \Big) \drm \varepsilon.
  \end{align*}
By computing an integration by parts and since $R(\underline \varepsilon) = R(\varepsilon_1 \vee \underline \varepsilon) = 0$ by assumption, we obtain:
\begin{align*}
    {\rm D} F_1(Q) (R)
    = &- \int_{\underline \varepsilon}^{\varepsilon_1 \vee \underline \varepsilon} R(\varepsilon)
    \Big(  \beta \big(\beta \varepsilon^2 Q''(\varepsilon) - 2 \big) \erm^{\beta (Q(\varepsilon) - \varepsilon Q'(\varepsilon))}
    + \underline\omega \big( 1 + \varepsilon Q''( \varepsilon ) \big)
    \erm^{Q'( \varepsilon )} \Big) \drm \varepsilon.
\end{align*}
Therefore, the Euler-Lagrange equation associated with the optimisation problem on $[\underline \varepsilon, \varepsilon_1 \vee \underline \varepsilon]$ is equivalent to the following non-linear second-order ODE:
\begin{align}\label{eq:ODE_1}
    0 = \underline\omega \big( 1 + \varepsilon Q''( \varepsilon ) \big) \erm^{Q'( \varepsilon )}
    +  \beta \big( \beta  \varepsilon^2 Q''(\varepsilon) - 2 \big)
    \erm^{\beta (Q(\varepsilon) - \varepsilon Q'(\varepsilon))}.
\end{align}
Moreover, we can compute the second derivative of $g_1$ with respect to $\eta$:
\begin{align*}
    g_1 ''(\eta) = \int_{\underline \varepsilon}^{\varepsilon_1 \vee \underline \varepsilon} \Big(
    \beta^2 \big( R(\varepsilon)  - \varepsilon R'(\varepsilon) \big)^2
    \erm^{\beta (Q(\varepsilon) + \eta R(\varepsilon) - \varepsilon Q'(\varepsilon) - \varepsilon \eta R'(\varepsilon))} 
    + \varepsilon \underline \omega \big(R'(\varepsilon)\big)^2 \erm^{Q'(\varepsilon) + \eta R'(\varepsilon)} 
    \Big) \drm \varepsilon.    
\end{align*}
This second derivative is therefore positive for any $\eta \in \R$ and implies that $F_1$ attains a minimum for $Q$ the solution on $[\underline \varepsilon, \varepsilon_1 \vee \underline \varepsilon]$ of ODE \eqref{eq:ODE_1}, if it exists.

\medskip

Similarly, to study the problem on $(\varepsilon_1 \vee \underline \varepsilon, 1)$, we consider $F_2$ defined by \eqref{eq:F_2}. Applying the same reasoning than for $F_1$, 
we obtain that the Euler-Lagrange equation associated with the optimisation problem on $[\underline \varepsilon, \varepsilon_1 \vee \underline \varepsilon]$ is equivalent to the following non-linear first-order ODE:
\begin{align}\label{eq:ODE_2}
    0 = \underline\omega (1+\alpha)^{\beta (\varepsilon_1 \vee \underline \varepsilon)+1} 
    \bigg( 1 + \varepsilon \frac{1+\alpha}{\alpha} Q''(\varepsilon) \bigg)
    \erm^{\frac{1+\alpha}{\alpha} Q' ( \varepsilon )} 
    + \beta \big( \beta  \varepsilon^2  Q''(\varepsilon)  - 2 \big)
    \erm^{\beta (Q(\varepsilon)- \varepsilon  Q' ( \varepsilon ))}.
\end{align}
Moreover, the second derivative of $g_2$ with respect to $\eta$ is positive for any $\eta \in \R$, which implies that $F_2$ attains a minimum for $Q$, the solution on $[\varepsilon_1 \vee \underline \varepsilon, 1]$ of ODE \eqref{eq:ODE_2}, if it exists.

\medskip

We can thus conclude that if there is a function $Q \in \Qc(\underline \varepsilon)$ solution to ODE \eqref{eq:ODE_1} on $[\underline \varepsilon, \varepsilon_1 \vee \underline \varepsilon]$ and to ODE \eqref{eq:ODE_2} on $[\varepsilon_1 \vee \underline \varepsilon, 1]$, it maximises the principal's profit for some $\underline \varepsilon \in [0,1]$ and $\underline q$ fixed. 
Combining both ODEs leads to ODE \eqref{eq:ODE}, which proves the theorem. Moreover, using \eqref{eq:profit_moche}, we obtain the form of the principal's profit given in \Cref{cor:pb_principal_SB}.
\end{proof}

\begin{remark}\label{rk:ODE_constraint}
As explained in \textnormal{\Cref{rk:ODE_necessary}}, \textnormal{\Cref{thm:sol_principal_ODE}} only gives a sufficient condition for the principal's optimisation problem. To obtain a necessary condition, one should adapt the previous proof by writing the Euler-Lagrange equation for the problem with constraints. A new \textnormal{ODE} would then be obtained, and the existence of a solution to this \textnormal{ODE} would be equivalent to the existence of an optimal contract. Nevertheless, in the application developed in \textnormal{\Cref{sec:fuel_poverty_TB}}, solving the \textnormal{ODE \eqref{eq:ODE}} is sufficient since its solution has the required regularity.
Moreover, one can prove that the \textnormal{ODE} has a unique solution for $\underline \varepsilon$ bounded away from $0$, which confirms that the numerical scheme converges to the solution of the principal's problem in our application. More precisely, on $[\underline \varepsilon,  \varepsilon_1 \vee \underline \varepsilon)$, the \textnormal{ODE \eqref{eq:ODE}} can be written as a system of two first-order ODEs as follows:
\begin{align*}
\left\{
    \arraycolsep=1.5pt\def\arraystretch{1.8}
    \begin{array}{ll}
    Q'(\varepsilon) = \dfrac{\beta Q(\varepsilon) - R(\varepsilon)}{1+\beta \varepsilon}, \\
    R'(\varepsilon) = \dfrac{1+ \beta \varepsilon}{\varepsilon} \dfrac{1 - 2 \beta \erm^{R(\varepsilon)} / \underline \omega}{1+ \beta^2 \varepsilon \erm^{R(\varepsilon)} / \underline \omega}.
    \end{array}
\right.
\end{align*}
By the Cauchy-Lipschitz theorem, the second \textnormal{ODE} has a unique solution if $\underline \varepsilon \geq c > 0$. This solution is, in particular, bounded with bounded derivatives on the interval considered, which implies that the first \textnormal{ODE} also has a unique solution. The same reasoning can be applied on the interval $[\varepsilon_1 \vee \underline \varepsilon, 1]$.
\end{remark}

\end{appendices}

\end{document}